\documentclass[sigconf]{acmart}
\AtBeginDocument{%
  \providecommand\BibTeX{{%
    \normalfont B\kern-0.5em{\scshape i\kern-0.25em b}\kern-0.8em\TeX}}}

\acmConference[Conference acronym 'XX]{Make sure to enter the correct
  conference title from your rights confirmation emai}{June 03--05,
  2018}{Woodstock, NY}
\usepackage{adjustbox}
\usepackage{array, boldline, makecell, booktabs}
\usepackage{dsfont}
\usepackage[mathscr]{euscript}
\usepackage{multirow}
\usepackage{subfigure}
\usepackage{caption}
\usepackage{tabularx}
\usepackage[boxed,ruled,commentsnumbered]{algorithm2e}
\usepackage{xcolor}

\newtheorem{prop}{Proposition}

\begin{document}

%%
%% The "title" command has an optional parameter,
%% allowing the author to define a "short title" to be used in page headers.
\title{From Fake News to \#FakeNews: Mining Direct and Indirect Relationships among Hashtags for Fake News Detection}

%%
%% The "author" command and its associated commands are used to define
%% the authors and their affiliations.
%% Of note is the shared affiliation of the first two authors, and the
%% "authornote" and "authornotemark" commands
%% used to denote shared contribution to the research.
\author{Xinyi Zhou}
\email{zhouxinyi@data.syr.edu}
\orcid{0000-0002-2388-254X}
\affiliation{%
  % \department{Data Lab, EECS Department}
  \institution{Syracuse University}
  % \streetaddress{P.O. Box 1212}
  % \city{Dublin}
  % \state{Ohio}
  \country{U.S.A}
  % \postcode{43017-6221}
}

\author{Reza Zafarani}
\email{reza@data.syr.edu}
\affiliation{%
  % \department{Data Lab, EECS Department}
  \institution{Syracuse University}
  % \streetaddress{P.O. Box 1212}
  % \city{Dublin}
  % \state{Ohio}
  \country{U.S.A}
  % \postcode{43017-6221}
}

\author{Emilio Ferrara}
\email{emiliofe@usc.edu}
\affiliation{%
  \institution{University of Southern California}
  % \streetaddress{P.O. Box 1212}
  % \city{Dublin}
  % \state{Ohio}
  \country{U.S.A}
  % \postcode{43017-6221}
}

\begin{abstract}
The COVID-19 pandemic has gained worldwide attention and allowed fake news, such as ``COVID-19 is the flu,'' to spread quickly and widely on social media. Combating this coronavirus infodemic demands effective methods to detect fake news. To this end, we propose a method to infer news credibility from hashtags involved in news dissemination on social media, motivated by the tight connection between hashtags and news credibility observed in our empirical analyses. We first introduce a new graph that captures all (direct and \textit{indirect}) relationships among hashtags. Then, a language-independent semi-supervised algorithm is developed to predict fake news based on this constructed graph. This study first investigates the indirect relationship among hashtags; the proposed approach can be extended to any homogeneous graph to capture a comprehensive relationship among nodes. Language independence opens the proposed method to multilingual fake news detection. Experiments conducted on two real-world datasets demonstrate the effectiveness of our approach in identifying fake news, especially at an \textit{early} stage of propagation.
\end{abstract}

%%
%% The code below is generated by the tool at http://dl.acm.org/ccs.cfm.
%% Please copy and paste the code instead of the example below.
%%
\begin{CCSXML}
<ccs2012>
   <concept>
       <concept_id>10003120.10003130</concept_id>
       <concept_desc>Human-centered computing~Collaborative and social computing</concept_desc>
       <concept_significance>500</concept_significance>
       </concept>
 </ccs2012>
\end{CCSXML}

\ccsdesc[500]{Human-centered computing~Collaborative and social computing}

%%
%% Keywords. The author(s) should pick words that accurately describe
%% the work being presented. Separate the keywords with commas.
\keywords{Fake news, information credibility, social media}

% %% A "teaser" image appears between the author and affiliation
% %% information and the body of the document, and typically spans the
% %% page.
% \begin{teaserfigure}
%   \includegraphics[width=\textwidth]{sampleteaser}
%   \caption{Seattle Mariners at Spring Training, 2010.}
%   \Description{Enjoying the baseball game from the third-base
%   seats. Ichiro Suzuki preparing to bat.}
%   \label{fig:teaser}
% \end{teaserfigure}

% \received{20 February 2007}
% \received[revised]{12 March 2009}
% \received[accepted]{5 June 2009}

%%
%% This command processes the author and affiliation and title
%% information and builds the first part of the formatted document.
\maketitle

\begin{figure}[t]
    \centering
    \includegraphics[width=0.85\columnwidth]{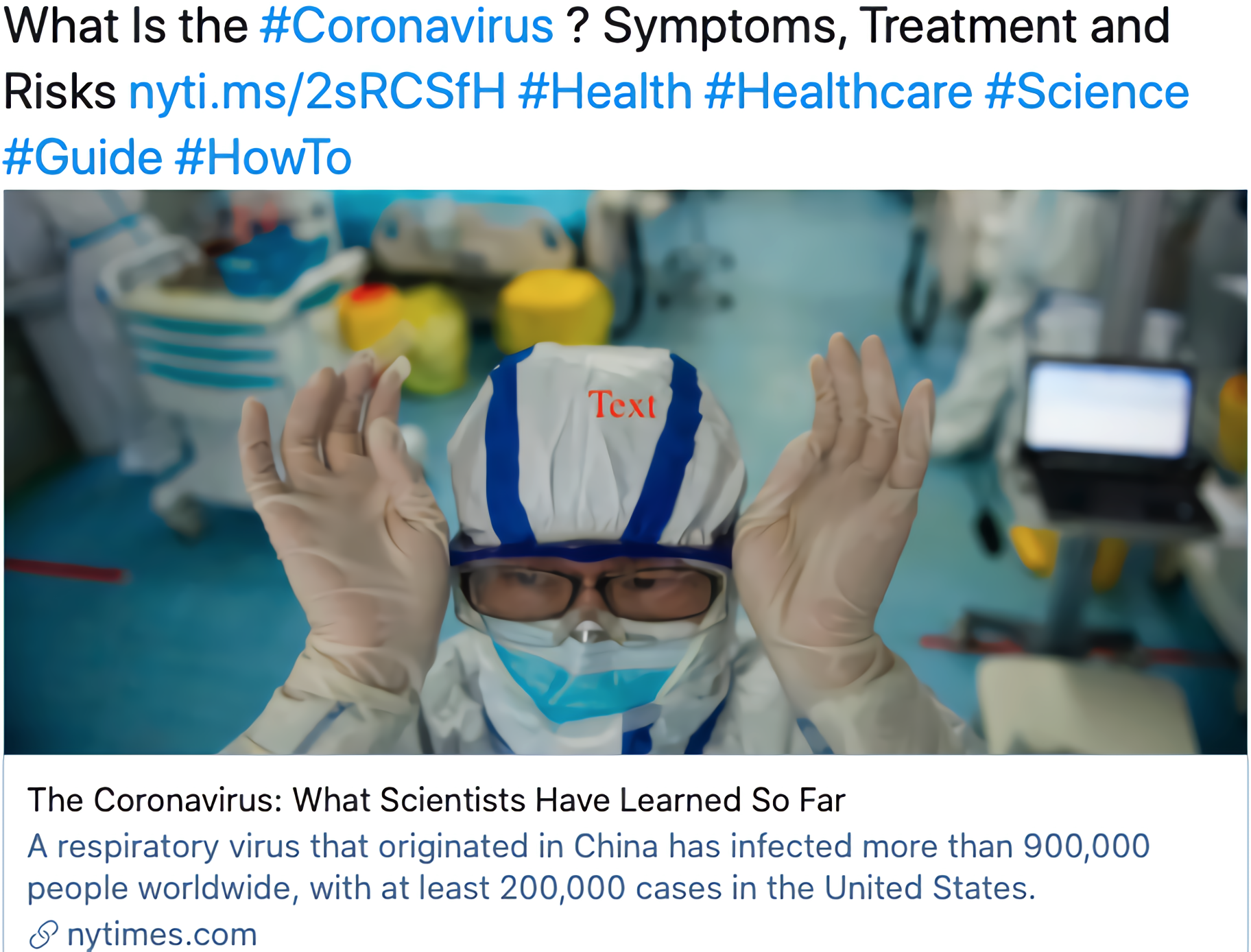}
    \includegraphics[width=0.85\columnwidth]{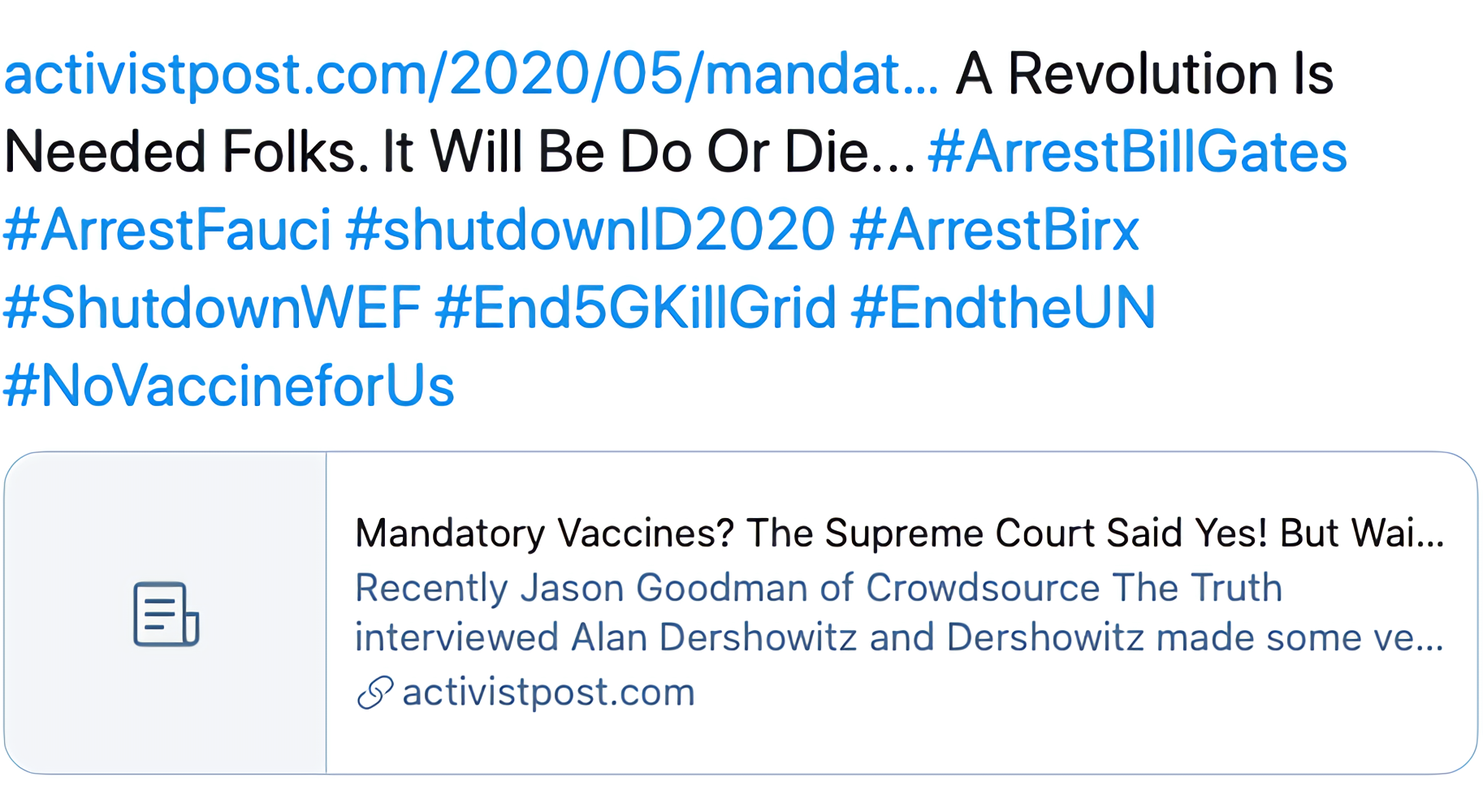} 
    \caption{Examples of User Posts. Hashtags can be a summarization of posts in topics, e.g., \#Coronavirus, \#Healthcare, and \#Science (top), and user opinions, e.g., \#ArrestBillGates \#ArrestFauci, and \#End5GKillGrid (bottom).}
    \label{fig:tweet_examples}
\end{figure}

\section{Introduction}
\label{sec:intro}
Evidence indicates that more than two-thirds of Americans get their news from social media~\cite{matsa2018news}. Such increased social media usage has created favorable conditions for the spread of fake news, which can have a detrimental impact on democracies, economies, and public health. For instance, extensive discussions have been sparked on whether fake news has affected the result of the Brexit vote or the 2016 U.S. presidential election has sparked discussions~\cite{allcott2017social}.
The hoax over the ``dead'' Ethereum founder spurs \$4 billion~\cite{roberts2017hoax}, and over `injured' Barack Obama wiped out \$130 billion in stock value~\cite{matthews2013does}. During the coronavirus pandemic, hundreds of news websites have peddled hoaxes;\footnote{\url{https://www.newsguardtech.com/coronavirus-misinformation-tracking-center/}} individuals who believe that drinking chlorine dioxide can cure or prevent coronavirus might take dangerous action to ``protect'' themselves from the virus. 

Combating fake news not only demands manual fact-checking by domain experts but also the development of algorithms for automatic verification, to which some existing studies have contributed significantly. These studies have investigated the role of news content~\cite{perez2018automatic,potthast2018stylometric,przybyla2020capturing,karimi2019learning,rubin2015towards,wang2018eann,zhang2020fakedetector,rashkin2017truth,bal2020analysing} and 
\textit{propagation information} (how news articles from media organizations or statements from public figures are spread on social media) in predicting fake news. Such propagation information has included post content that often contains users' sentiments or stances~\cite{jin2016news}, post relationships like forwarding and replies~\cite{wu2015false,lukasik2019gaussian,ma2018rumor,zhou2022fake}, user profiles~\cite{cheng2021causal,shu2019role}, and social connections~\cite{wu2018tracing}. 
At times, hashtags have also been utilized by, e.g., looking at whether they exist or how many they are in each user post~\cite{castillo2011information}. Nevertheless, hashtag data are significantly underinvestigated in evaluating news credibility, and the \textit{relationship} among hashtags has been overlooked; their values can be multi-fold.

First, hashtag data is insensitive to privacy and bias issues especially compared to the user profile data. Second, hashtags contain intuitive and plentiful information, such as providing topic and stance (see Figure \ref{fig:tweet_examples} for examples), and can act as an informative summary of user posts. Hence, topic modeling and stance extraction methods introducing additional inference are not required, whereas they are often necessary for the general post content. Such advantages in hashtags have formed the foundation for research on learning representations for posts, where hashtags often act as the ground truth in learning~\cite{dhingra2016tweet2vec}.

Furthermore, we observe a tight connection between news credibility and hashtags in the news spreading. Precisely, fake news significantly corresponds to the hashtag only used to spread fake news (see Figure \ref{fig:news_hashtags_fake}). Likewise, true news primarily uses hashtags only attached to posts spreading true news (see Figure \ref{fig:news_hashtags_true}). Therefore, mining the relationship between news and hashtags can benefit news credibility inference. Mining the (direct and indirect) relationship among hashtags can further help with it.

\begin{figure}[t]
\begin{minipage}{\columnwidth}
    \subfigure[Fake News]{\label{fig:news_hashtags_fake}
        \centering
        \includegraphics[width=0.475\columnwidth]{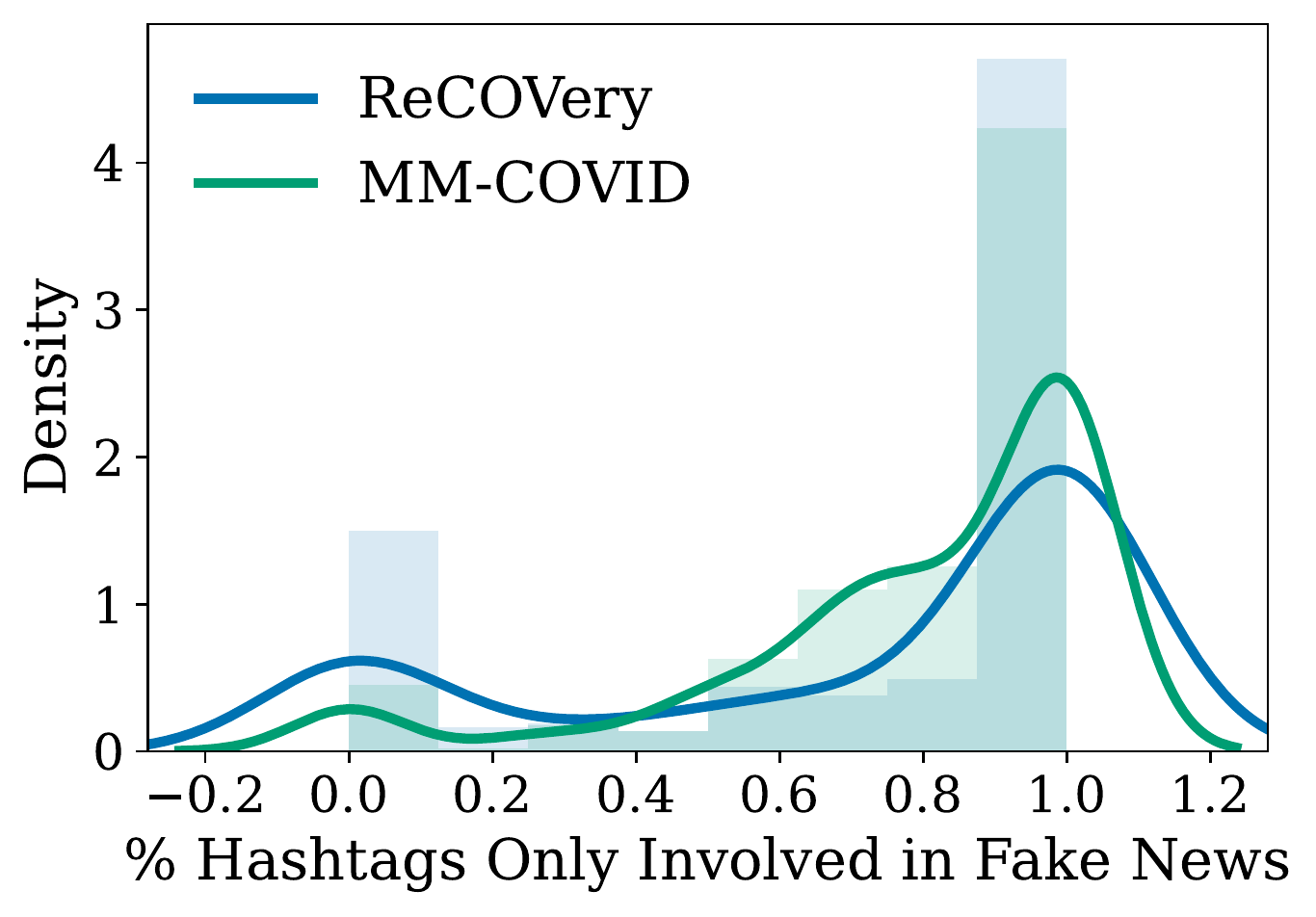}
    }
    \subfigure[True News]{\label{fig:news_hashtags_true}
        \centering
        \includegraphics[width=0.485\columnwidth]{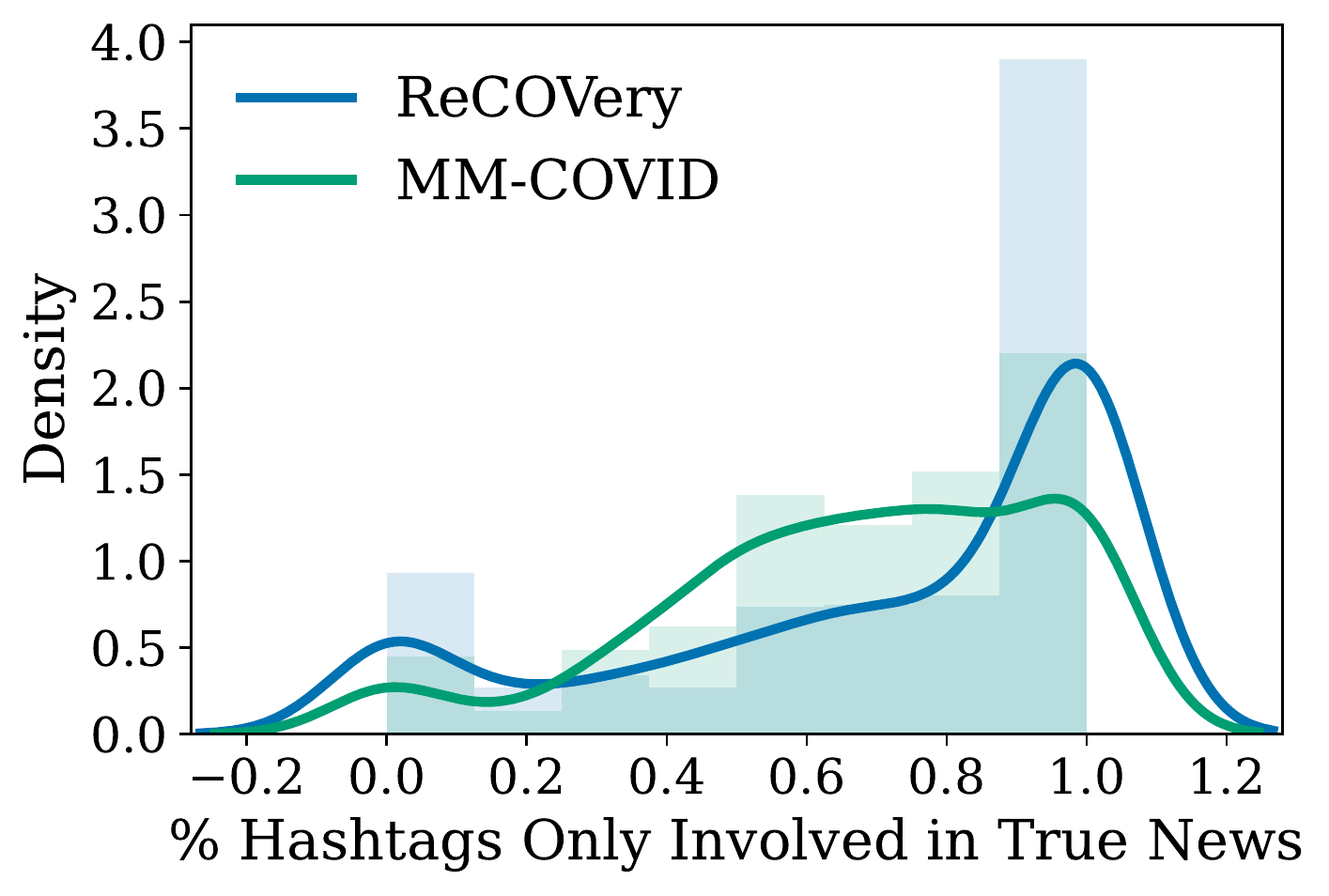}
    }
\end{minipage}
\caption{Connection between News Credibility and Hashtags in News Propagation on Social Media. For each news article or statement, we compute the proportion of hashtags within its corresponding hashtag set involved in spreading either fake or true news. Distributions indicate that (a) fake news significantly corresponds to the hashtag only used to spread fake news. Likewise, (b) true news mainly uses hashtags merely attached to posts spreading true news.}
\label{fig:news_hashtags}
\end{figure}

On the other hand, there are still open issues on how to effectively detect fake news (i) with limited ground-truth data, as news credibility annotation is time-consuming and requires domain knowledge; (ii) in its early dissemination as individuals tend to trust fake news stories more after they have been widely spread (i.e., \textit{validity effect}~\cite{hogg2020social}); and (iii) across languages as fake news has become a universal problem.

We aim to address these open issues in automatic fake news detection by incorporating valuable hashtag information in news credibility evaluation. To this end, we propose a method to identify fake news by mining a newly proposed graph that captures all relationships among hashtags, including direct and \textit{indirect}. Then, a semi-supervised algorithm is developed to predict fake news based on the constructed graph; the method is independent of news languages. The overall contributions of this work are summarized as follows:
\begin{enumerate}
    \item We first use hashtags as a strong and informative proxy to detect fake news. The proposed approach uses an iterative solution to compute direct and indirect relationships among hashtags, which can be extended to any homogeneous graph to capture a comprehensive relationship among nodes.

    \item We develop a language-independent semi-supervised algorithm. The algorithm allows to predict fake news using the limited number of labeled multilingual news data.
    
    \item We conduct extensive experiments based on two real-world datasets. The results demonstrate the effectiveness of the proposed method in identifying fake news, especially at an early stage of propagation. 
\end{enumerate}
% The source files are available at \url{https://anonymous.4open.science/r/NewsTag-CD0F/} for reproduction.

The rest of the paper is organized as follows. We introduce and evaluate the proposed method in Sections \ref{sec:method}--\ref{sec:experiment}, review related work in Section \ref{sec:review}, and conclude with future directions in Section \ref{sec:conclusion}.

\begin{figure*}[t]
    \centering
    \includegraphics[width=\textwidth]{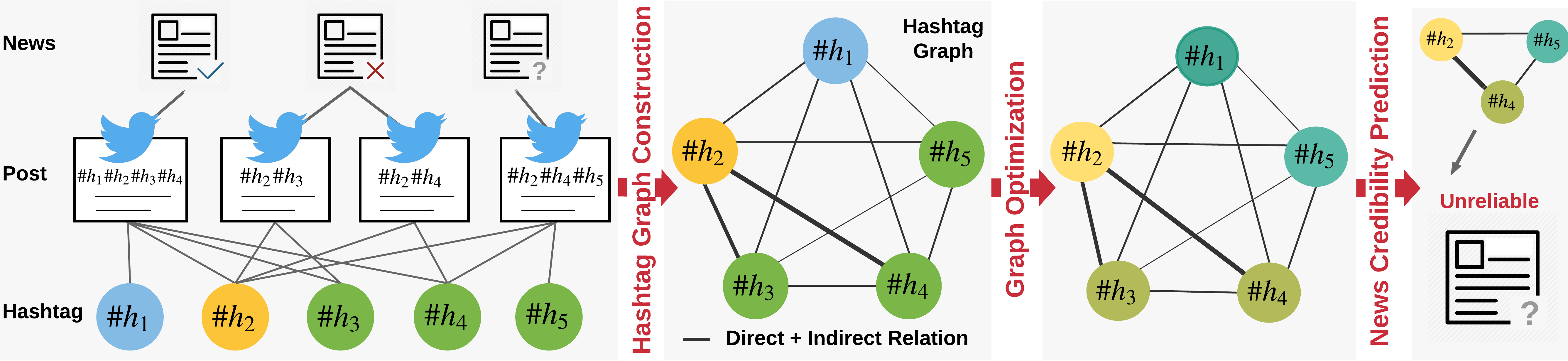}
    \caption{Method Overview.}
    \label{fig:overview}
\end{figure*}

\section{Methodology}
\label{sec:method}

We specify the proposed method by first introducing a novel graph that captures all (direct and \textit{indirect}) relationships among hashtags involved in news propagation on social media (see Section \ref{subsec:hashtag_graph}). Then, we detail the proposed algorithm based on the constructed graph (see Section \ref{subsec:cred_propagation}).
Figure \ref{fig:overview} provides an overview of the proposed method.

\vspace{1em}
\noindent\textbf{Problem Definition and Notations}\quad
Given $m$ news articles or statements $A = \{ a_1, a_2, \cdots,$ $a_{m'}, a_{m'+1}, \cdots, a_m \}$, the credibility of $m'$ of them is available and denoted as $Y = \{y_1, y_2, \cdots, y_{m'}: y_i\in \{-1,1\},  i\leq m' \}$, where $-1$ indicates fake news and 1 denotes true news. For each news article or statement $a_i \in A$, suppose it has been shared by a total of $p_i$ posts, each of which contains a hashtag set $H^i_j, j\leq p_i$. 
Let $H = \{ h_1, h_2, \cdots, h_q\}$ be the overall hashtag set for $m$ news articles or statements, where $H = \bigcup_{
i=1}^{m}\bigcup_{j=1}^{p_i} H^i_j$. Our goal is to predict the credibility of news articles or statements $a_{m'+1}, \cdots, a_m$.

\subsection{All-relations Hashtag Graph}
\label{subsec:hashtag_graph}

Let $G = (V, E, W)$ be the proposed hashtag graph, where $V = H$ (i.e., the set of distinct hashtags for $m$ news articles or statements) is the node set. Two nodes are connected if and only if the hashtags represented by these two nodes are used within the same post, i.e., $(h_k, h_l)\in E~\text{\textit{iff.}}~\{h_k,h_l\}\subseteq H_j^i$. The weight of an edge between two nodes captures the co-occurrence among their corresponding hashtags, equivalent to the number of hashtag co-occurrences among all news posts. Formally, 
\begin{equation}
\label{eq:graph_weight}
w_{kl} = \sum_{i=1}^{m}\sum_{j=1}^{p_i}\mathcal{I}(\{h_k,h_l\}\subseteq H_j^i),
\end{equation}
where $\mathcal{I}(*) = 1$ if $*$ is true; otherwise, $\mathcal{I}(*) = 0$.

However, the hashtag graph defined thus far can only capture the \textit{direct} relation between each pair of hashtags. Intuitively, if $h_a$ and $h_b$ appear in the same post, and $h_b$ and $h_c$ simultaneously appear in the other one, $h_a$ and $h_c$ should share an \textit{indirect} relation to some extent; however, in the above-defined graph, it is possible that $(h_a,h_c)\notin E$. To address this issue, we develop the following method to capture all (direct+indirect) relations among hashtags.

Let $\mathbf{W}_{\textup{dir}} = [w_{kl}]_{k,l=1}^{q}$ denote the direct relation matrix of $G$. Note that $\mathbf{W}_{\textup{dir}} \geq 0$ and $\mathbf{W}_{\textup{dir}} = \mathbf{W}_{\textup{dir}}^\top$. We use $\mathbf{W}_{\textup{all}}$ to denote the all-relations matrix developed from $\mathbf{W}_{\textup{dir}}$, whose entries capture direct and indirect relations between two hashtags. $\mathbf{W}_{\textup{all}}$ is computed by 
\begin{equation}
\label{eq:totoal_relation_matrix}
    \begin{aligned}
    \mathbf{W}_{\textup{all}}  & = \lim_{k_1 \rightarrow \infty } (\mathbf{N}+\mathbf{N}^2+\cdots+\mathbf{N}^{k_1})  \\ 
     & =  \lim_{k_1 \rightarrow \infty } \mathbf{N} (\mathbf{I}-\mathbf{N}^{k_1})(\mathbf{I}-\mathbf{N})^{-1} \\ 
     & = \mathbf{N} (\mathbf{I}-\mathbf{N})^{-1},
    \end{aligned}
\end{equation}
where $\mathbf{I}$ is an identity matrix, and
\begin{equation}\label{eq:normalization}
\mathbf{N} = \mathbf{W}_{\textup{dir}} / \max (\sum_l w_{kl})
\end{equation}
to guarantee convergence. Note that consistent with the properties of $\mathbf{W}_{\textup{dir}}$, $\mathbf{W}_{\textup{all}} \geq 0$ and $\mathbf{W}_{\textup{all}} = \mathbf{W}_{\textup{all}}^\top$.

\subsection{Fake News Detection}
\label{subsec:cred_propagation}

This section introduces the proposed language-independent semi-supervised algorithm for fake news detection. Generally speaking, our task is to infer the credibility of news articles or statements $a_{m'+1}, a_{m'+2}, \cdots, a_m$ from $a_1, a_2, \cdots, a_{m'}$ whose credibility has been available and denoted as $y_1, y_2, \cdots, y_{m'}$. We use hashtags involved in news spreading and their (direct and indirect) relationships as an intermediary for this inference. 

Precisely, for each hashtag $h_k$, it initially corresponds to a credibility distribution of its involved news articles or statements whose credibility has been available. The distribution has a weighted average credibility score, which we compute as 
\begin{equation}
\label{eq:initial_cred}
\mathbf{c_0}_k = \sum_{i=1}^{m'} \sum_{j=1}^{p_i} y_i \mathcal{I}(h_k \in H_j^i) \big/ \sum_{i=1}^{m'} \sum_{j=1}^{p_i} \mathcal{I}(h_k \in H_j^i),
\end{equation}
The score $\mathbf{c_0}_k\in [-1,1]$ can be seen as the pseudo-ground-truth credibility of the hashtag $h_k$ obtained by the subset of news articles or statements whose credibility has been available. $\mathbf{c_0}_k=-1$ indicates $h_k$ is thus far only involved in the spread of unreliable news, while $\mathbf{c_0}_k=1$ indicates that it thus far only corresponds to reliable news. It is possible that for some of the hashtags in $H$, their involved news credibility is unknown to us at this stage, which we initially set as zero. 

Next, we optimize these initial ``credibility'' scores of hashtags based on the structure of the constructed all-relations hashtag graph. Assuming that two hashtags having a closer (direct and indirect) relationship share a more similar ``credibility'' score, we define our cost function as follows~\cite{zhou2004learning}.
\begin{equation}
\label{eq:cred_propagation}
    \begin{aligned} 
    \mathcal{L}(\mathbf{c}) & = \mu \sum\limits_{k=1}^{q} \sum\limits_{l=1}^{q} \mathbf{W}_{\textup{all}_{kl}}(\frac{\mathbf{c}_k}{\sqrt{\mathbf{D}_{kk}}} - \frac{\mathbf{c}_l}{\sqrt{\mathbf{D}_{ll}}})^2 + (1-\mu) \sum\limits_{k=1}^{q} (\mathbf{c}_k - \mathbf{c_{0}}_k)^2 \\
    & = \mu \left \|\mathbf{W}_{\textup{all}} \circ \mathbf{Y} \right \|_1 + (1-\mu) \left \| \mathbf{c} - \mathbf{c_0} \right \|_2,
    \end{aligned}
\end{equation}
where 
$\mathbf{Y} = (\mathbf{D}^{-1/2}\mathbf{C}^\top - \mathbf{C}\mathbf{D}^{-1/2})^{\circ 2}$, 
$\mathbf{C} = \mathbf{1}\mathbf{c}^\top$, 
$\mathbf{D}$ is a diagonal matrix with $\mathbf{D}_{kk} = \sum_{l} \mathbf{W}_{\textup{all}_{kl}}$, and 
$\mu \in (0, 1)$ is a regularization parameter.

Note that the cost function $\mathcal{L}$ is convex for $\mathbf{W}_{\textup{all}} \geq 0$. Therefore, the optimized hashtag ``credibility'', which we aim to keep consistent with the hashtag relationship, is estimated by
\begin{equation}
\begin{array}{rcl}
    \mathbf{\hat{c}} = \arg \min\limits_{\mathbf{c}}\mathcal{L}(\mathbf{c}) 
    & \Leftrightarrow & \frac{\partial \mathcal{L}}{\partial \mathbf{c}}\big|_{\mathbf{c}=\mathbf{\hat{c}}} = 0 \\
    & \Leftrightarrow & \mu \mathbf{\hat{c}}(\mathbf{E-X}) + (1-\mu)(\mathbf{\hat{c}-\mathbf{c_0}}) = 0 \\
    & \Leftrightarrow & \mathbf{\hat{c}} = (1-\mu) (\mathbf{E}-\mu\mathbf{X})^{-1}\mathbf{c_0},
\end{array}
\end{equation}
where $\mathbf{X} = \mathbf{D}^{-1/2}\mathbf{W}_{\textup{all}}\mathbf{D}^{-1/2}$.

To avoid the computation of matrix inversion whose exact solution might not exist, we can get $\mathbf{\hat{c}}$ iteratively with the following update rule (see the proof in Proposition~\ref{prop:iteration}):
\begin{equation}
\label{eq:k2}
    \mathbf{c}(k_2) = \mu\mathbf{X}\mathbf{c}(k_2-1) + (1-\mu)\mathbf{c_0},
\end{equation}
where $k_2=1,2,\cdots$ and $\mathbf{c}(0) = \mathbf{c_0}$.

\begin{prop}
\label{prop:iteration}
$\mathbf{\hat{c}} = \lim\limits_{k_2 \rightarrow \infty } \mathbf{c}(k_2)$.
\end{prop}
\begin{proof}
$\lim\limits_{k_2 \rightarrow \infty } \mathbf{c}(k_2)
 = (\mu\mathbf{X})^{k_2}\mathbf{c}_0 + (1-\mu) \mathbf{c}_0 (\mathbf{E} - (\mu \mathbf{X})^{k_2})(\mathbf{E}-\mu \mathbf{X})^{-1} 
 = (1-\mu) (\mathbf{E} - \mu \mathbf{X})^{-1} \mathbf{c_0} = \mathbf{\hat{c}}$
\end{proof}

Finally, for each news article or statement $a_i (m'+1\leq i \leq m)$ whose credibility has not been available, we can estimate its credibility, denoted as $\hat{y}_i$, based on the credibility distribution of hashtags involved in its spread on social media. Formally, 
\begin{equation}
\label{eq:y_pred}
    \hat{y}_i = \left\{\begin{array}{rl}
    1, & \text{if}~\sum\limits_{j=1}^{p_i} \sum\limits_{k=1}^{q} \mathbf{\hat{c}}_k \mathcal{I}(h_k \in H_j^i)>0; \\ 
    -1, & \text{otherwise},
\end{array}\right.
\end{equation}
$\hat{y}_i=1$ indicates that news $a_i$ is predicted as true news and $\hat{y}_i=-1$ means the news is predicted to be fake.

\vspace{2mm}
We summarize the method, named NewsTag, in Algorithm \ref{alg:newshat}.

\begin{algorithm}[t]
\LinesNumbered
\caption{NewsTag}
\label{alg:newshat}
\KwIn{$A=\{a_i\}_{i=1}^{m}$, $Y=\{y_{i_1}\}_{i_1=1}^{m'}$, $\{H_j^i \}_{i,j=1}^{i=m, j=p_i}$, $H=\{h_k\}_{k=1}^{q}$ and $\mu$, where $y_{i_1}\in\{-1,1\}$, $m'<m$, $H=\bigcup_{
i=1}^{m}\bigcup_{j=1}^{p_i} H^i_j$, and $\mu\in(0,1)$}
\KwOut{$\{\hat{y}_{i_2}\}_{i_2=m'+1}^{m}$, where $\hat{y}_{i_2}\in\{-1,1\}$}

\# Construct the hashtag graph $G=\{H,E,\mathbf{W}_{\textup{all}}\}$\;  
\ForEach{k,l}{
    $e_{kl}\in E$ \text{iff} $\{h_k,h_l\}\subseteq H_j^i$\;
    $\mathbf{W}_{\textup{dir}_{kl}} = \sum_{i=1}^{m}\sum_{j=1}^{p_i}\mathcal{I}(\{h_k,h_l\}\subseteq H_j^i)$\;}
    
$\mathbf{N} = \mathbf{W}_{\textup{dir}} / \max ([\sum_{l} w_{kl}]_{k=1}^q)$ (Eq. (\ref{eq:normalization})); \# Normalization\;
Initialize $\mathbf{W}_{\textup{all}} = [\mathbf{0}]_{q\times q}$ and $\mathbf{N'} = \mathbf{N}$\;
\While{not convergence}{
    $\mathbf{N'} \leftarrow \mathbf{N'}\times \mathbf{N}$\;
    $\mathbf{W}_{\textup{all}} \leftarrow \mathbf{W}_{\textup{all}} + \mathbf{N'}$\;
}

\# Predict news credibility based on the graph\;
Initialize $\mathbf{D}=[\mathbf{0}]_{q\times q}$\;
\ForEach{k}{
    {$\mathbf{c_0}_k = \sum_{i_1=1}^{m'} \sum_{j=1}^{p_{i_1}} y_{i_1} \mathcal{I}(h_k \in H_j^{i_1}) \big/ \sum_{i_1=1}^{m'} \sum_{j=1}^{p_{i_1}} \mathcal{I}(h_k \in H_j^{i_1})$\;}}
    $\mathbf{D}_{kk} = \sum_{l} \mathbf{W}_{\textup{all}_{kl}}$\;
$\mathbf{X} = \mathbf{D}^{-1/2}\mathbf{W}_{\textup{all}}\mathbf{D}^{-1/2}$\;

Initialize $\mathbf{c}=\mathbf{c_0}$\;
\While{not convergence}{
    $\mathbf{c} \leftarrow \mu\mathbf{X}\mathbf{c} + (1-\mu)\mathbf{c_0}$\;
}

\ForEach{$i_2$}{
    $\hat{y}_{i_2} = 1$ if $\sum_{j=1}^{p_{i_2}} \sum_{k=1}^{q} \mathbf{c}_k \mathcal{I}(h_k \in H_j^{i_2})>0$ else $=-1$\;} 
    
\Return{$\{\hat{y}_{i_2}\}_{i_2=m'+1}^{m}$}
\end{algorithm} 

\section{Experiments}
\label{sec:experiment}

Our experimental setup is introduced in Section \ref{subsec:setup}, followed by the experimental results in Section \ref{subsec:results}. 

\subsection{Experimental Setup}
\label{subsec:setup}

We detail the data used in the experiments (see Section \ref{subsubsec:datasets}), the baselines that the proposed method is compared to (see Section \ref{subsubsec:baselines}), and the implementation details (see Section \ref{subsubsec:implementation_details}).

\subsubsection{Datasets}
\label{subsubsec:datasets}

Experiments are conducted on the ReCOVery~\cite{zhou2020recovery} and MM-COVID~\cite{li2020mm} datasets. Both datasets collect verified news statements or articles about the coronavirus (i.e., content information) and their dissemination on Twitter (social context information). ReCOVery data focus on English news articles while corresponding tweets can be non-English. MM-COVID data contain multilingual statements of news.
Data statistics are in Table \ref{tab:data_statistics}.

\begin{table}[t]
    \centering
    \captionof{table}{Data Statistics.}
    \label{tab:data_statistics}
    \begin{tabular}{lrr}
    \toprule[1pt]
                                 & \textbf{ReCOVery} & \textbf{MM-COVID} \\ \midrule[0.5pt]
    \textbf{\# News Articles}             & 1,768             & 1,054\\
    \quad \textbf{- Reliable News}        & 1,231             & 473  \\
    \quad \textbf{- Unreliable News}      & 537               & 581  \\
    \textbf{\# Distinct Hashtags ($|V|$)} & 3,433             & 7,352\\
    \textbf{\# Hashtag-Hashtag ($|E|$)}   & 5,923             & 30,847\\
    \textbf{\# Tweets}                   & 45,702            & 143,988\\
    \bottomrule[1pt]
    \end{tabular}
\end{table}

\subsubsection{Baselines}
\label{subsubsec:baselines}

We compare the proposed method (NewsTag) with six baselines that can be categorized into three groups:

\paragraph{I. Content-based methods} These are LIWC~\cite{pennebaker2015development} and RST~\cite{ji2014representation}. Both approaches examine the role of \textit{news content} in predicting fake news. Differently, LIWC represents the writing style of news articles or statements at the semantic level, while RST is a discourse-level method.

\paragraph{II. Propagation-based methods} To compare the role of hashtags (i.e., the proposed approach) with that of \textit{user posts} in predicting fake news, we involve two methods (i.e., T2VC and T2VW)~\cite{dhingra2016tweet2vec}, which develop a Bi-GRU encoder to learn particularly the representation of user tweets that often contain informal expressions, spelling errors, and special characters. Differently, the embedding of tweets in T2VC is at the character level, while T2VW is a word-level method.

\paragraph{III. Variants of the proposed method} 
To validate the indirect relationship among hashtags, we compare the proposed method with NewsTag$\backslash$I, which stands for the NewsTag method that merely considers the direct relationship among hashtags. Meanwhile, we involve NewsTag$\backslash$W as one baseline, which stands for the unweighted NewsTag method. This unweighted version includes the unweighted graph and, correspondingly, simplified Eq.~(\ref{eq:initial_cred}) (as $\mathbf{c_0}_k = \sum_i y_i \mathcal{I}(h_k \in H_i) \big/ \sum_i \mathcal{I}(h_k \in H_i)$) and Eq.~(\ref{eq:y_pred}) (as $\hat{y}_i = 1~\text{if}~\sum_{k} \mathbf{\hat{c}}_k$ $\mathcal{I}(h_k \in H_i)>0$, otherwise -1), where $H_i = \bigcup_j H^i_j$.

\subsubsection{Implementation Details}
\label{subsubsec:implementation_details}

News articles or statements in each dataset are randomly grouped as training data (80\%) and testing data (20\%). Macro F1 and micro F1 scores are used to evaluate method performance. 
For the proposed method, we use grid-search by varying $\mu$'s value from 0 to 1 with a step size of 0.1 to get the best performance, which is finally set as 0.4. We set $k_1 = 10$ and $k_2 = 5$, where both solutions converge (see Figure \ref{fig:convergence}).
For baselines, we experiment with various classifiers (decision trees, logistic regression, $k$-NN, random forest,  na\"ive Bayes, and SVM) with default hyperparameters; we present the best-performing one.

\begin{figure}[t]
    \subfigure[ReCOVery]{
        \includegraphics[width=0.48\columnwidth]{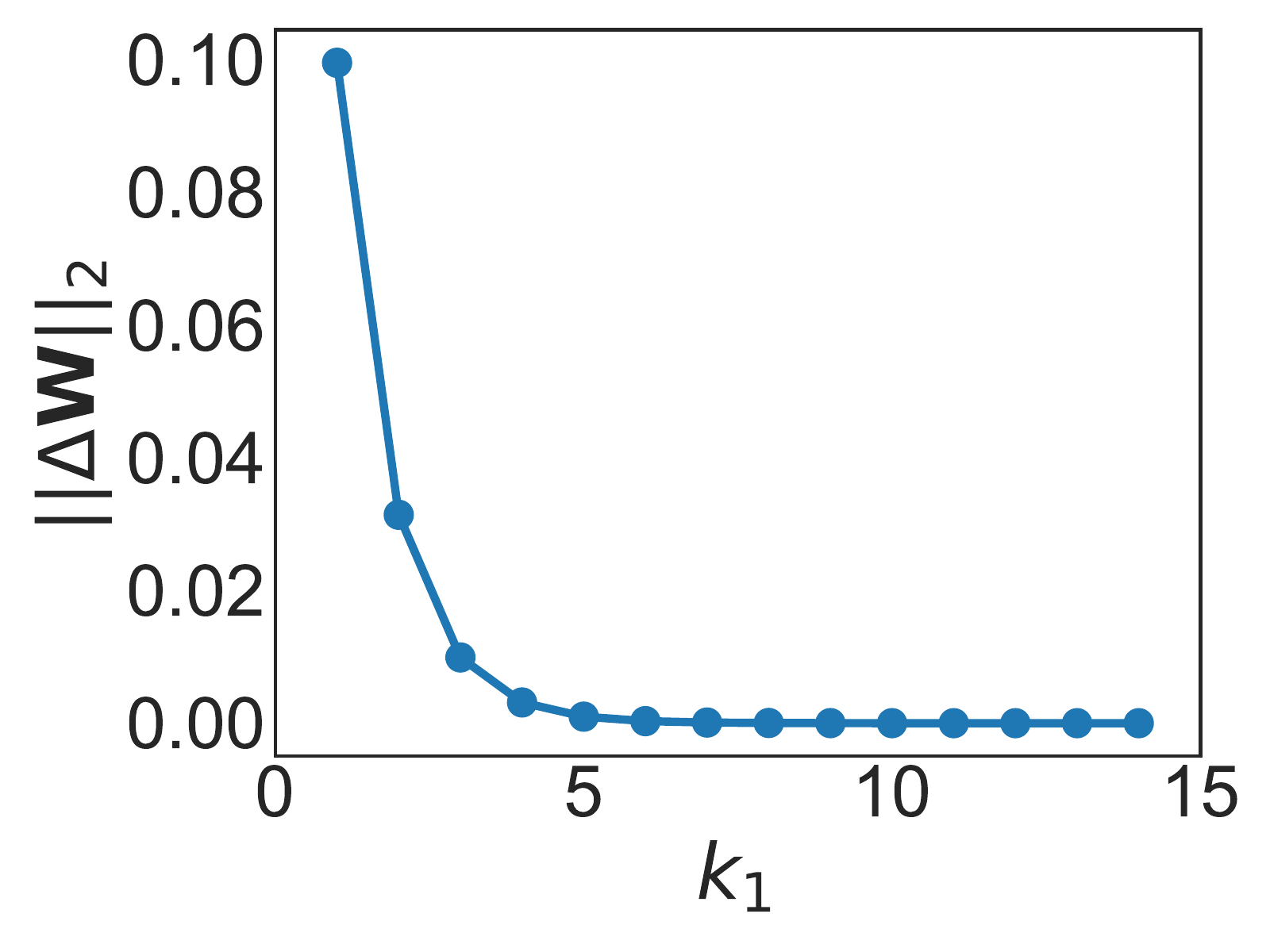}
        \includegraphics[width=0.49\columnwidth]{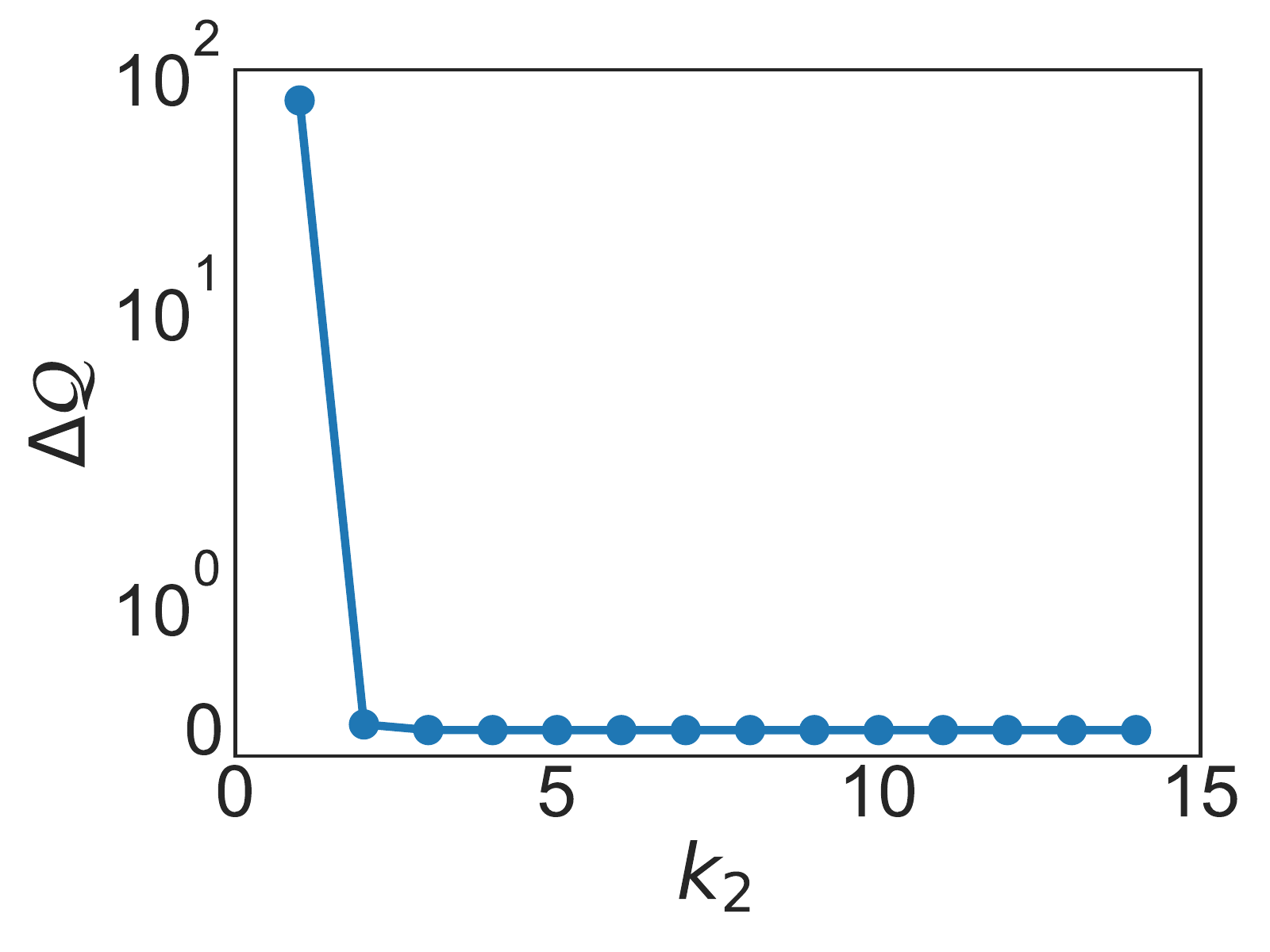}
    }
    \subfigure[MM-COVID]{
        \includegraphics[width=0.48\columnwidth]{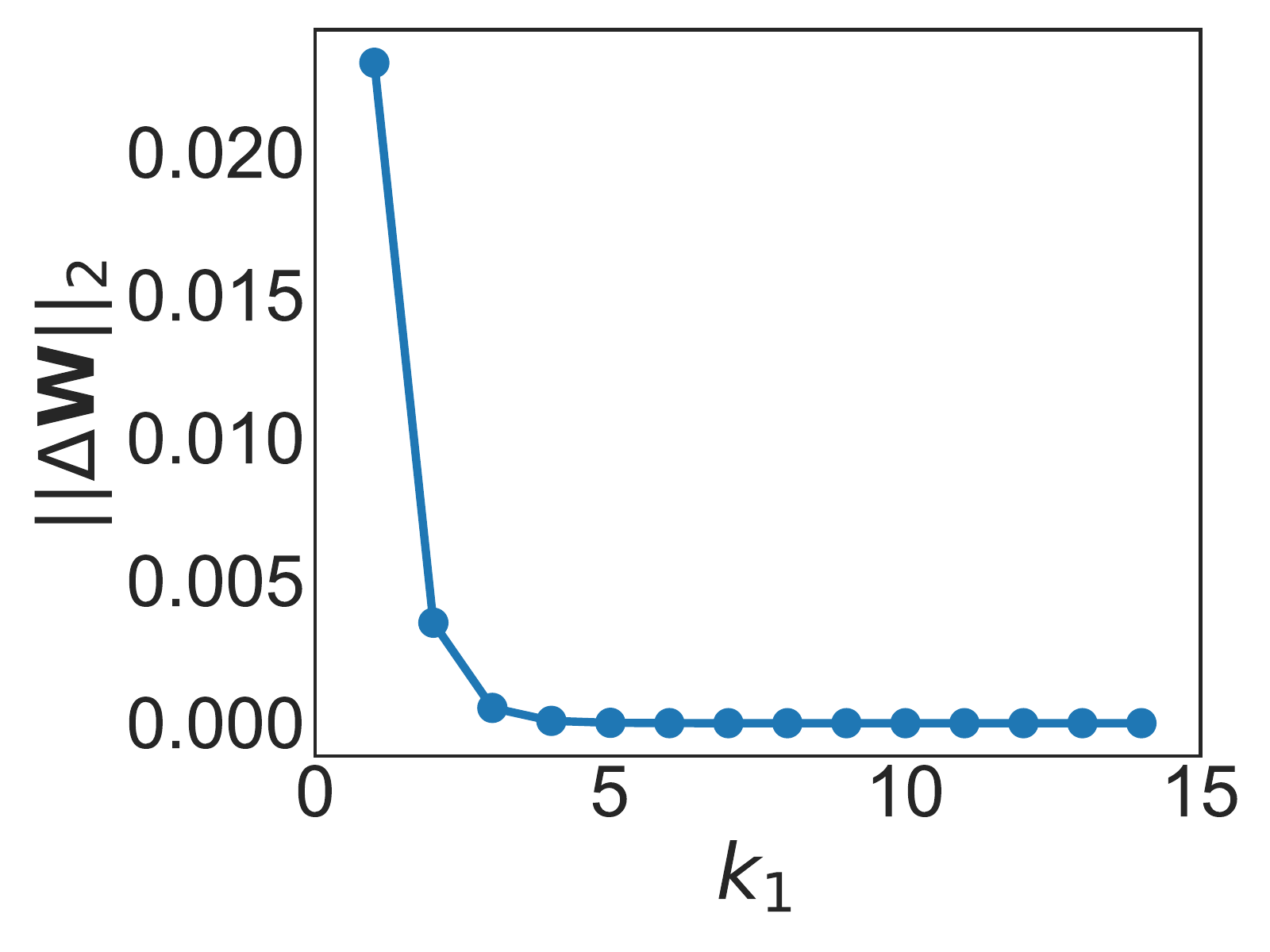}
        \includegraphics[width=0.49\columnwidth]{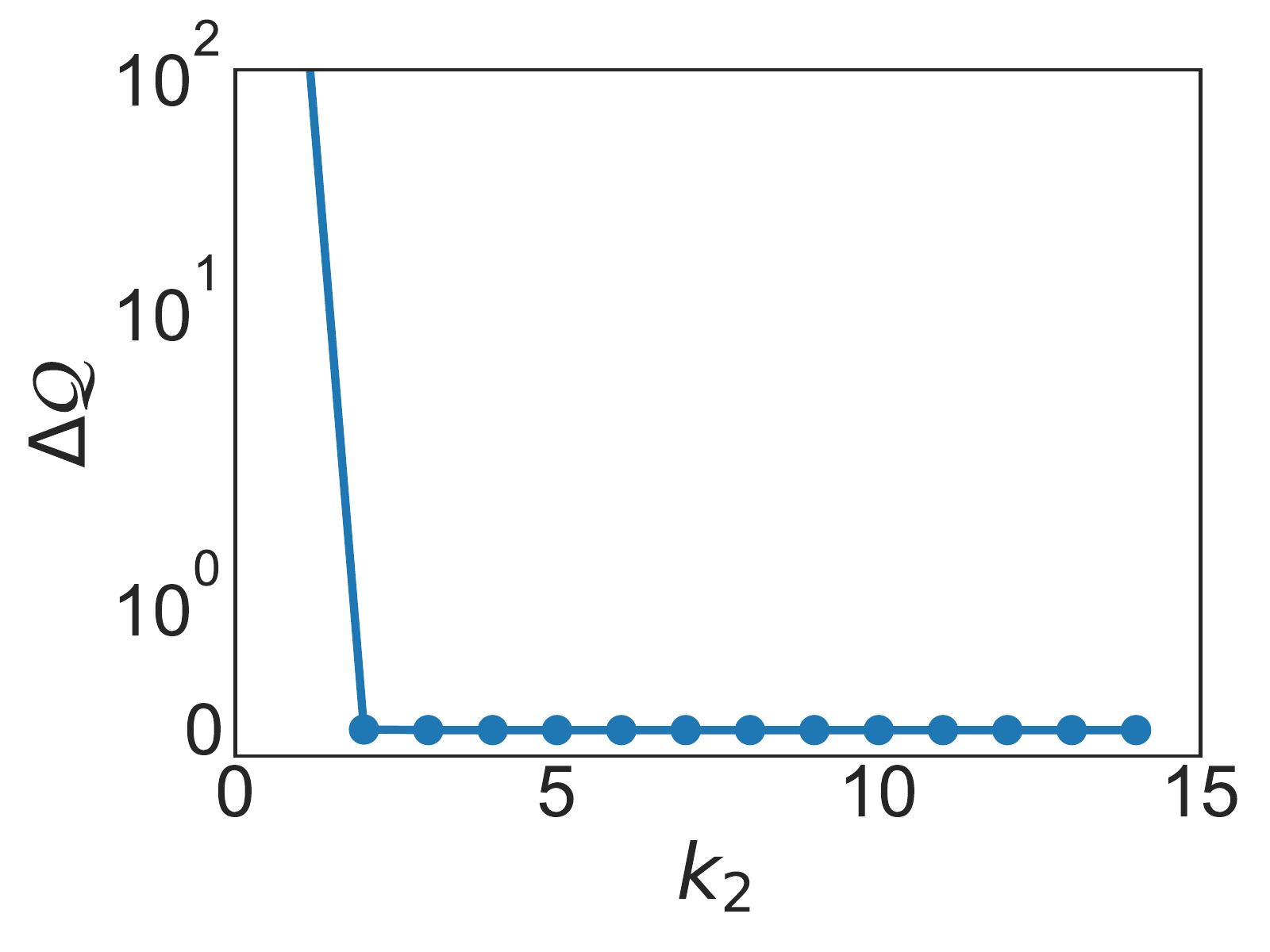}
    }    
    \captionof{figure}{Convergence Analysis.}
    \label{fig:convergence}
\end{figure}

\subsection{Experimental Results}
\label{subsec:results}

The proposed method is evaluated from the following three perspectives by comparing with our baselines: the general performance of methods (see Section \ref{subsubsec:general_performance}),  the performance of methods for fake news early detection (see Section \ref{subsubsec:early_detection}), and a case study (see Section \ref{subsubsec:case_study}).

\subsubsection{General Performance}
\label{subsubsec:general_performance}

Table~\ref{tab:general_performance} provides the results of the proposed method and baselines for predicting fake news. Results indicate that the method performance is $\text{NewsTag} > \text{NewsTag}\backslash\text{I} \approx \text{NewsTag}\backslash\text{W} > \text{T2VC} \approx \text{T2VW} > \text{LIWC} \approx \text{RST}$. The proposed method predicts fake news with around 0.73 macro F1 and 0.81 micro F1 scores based on ReCOVery data. For the MM-COVID dataset, the proposed method's macro and micro F1 scores are $\sim$0.85. For content-based methods, LIWC and RST, the proposed method can significantly outperform them by improving around 20\% on the macro and micro F1 scores on the ReCOVery dataset and over 15\% on both F1 scores on the MM-COVID dataset. Compared to propagation-based methods (T2VC and T2VW) that exploit post content in predicting fake news, the proposed method exploiting hashtag information can increase macro and micro F1 scores by around 10\% with ReCOVery data and 8\% with MM-COVID data. Compared to its two variants, the proposed method improves $\sim$3\% on both F1 scores in two datasets.

\vspace{1em}
\noindent \textbf{An Explanation for Method Performance}\quad
We explore the possible explanation for the better performance of the proposed approach, NewsTag, compared to its variants, NewsTag$\backslash$I and News- Tag$\backslash$W. For NewsTag$\backslash$I, it is comparatively intuitive that NewsTag helps improve upon NewsTag$\backslash$I by mining indirect relationships among hashtags on top of their direct relations. As a result, the hashtag graph becomes denser with cumulative weights and provides more information to infer news credibility (see Figure \ref{fig:evolutional_hashtag_graph} for an illustration). On the other hand, as NewsTag, compared to NewsTag$\backslash$W, further involves the number of tweets related to each news article (i.e., news popularity on social media) in inference, the question arises if such popularity differs between true news and fake news. To answer this question, we record the overall popularity of each news article or statement within fixed hours after it was published and draw the boxplot for true and fake news (see Figure \ref{fig:news_patterns}), respectively. Results indicate that the popularity of true news is generally higher than fake news for both ReCOVery and MM-COVID data.

\begin{table}[t]
\centering
\caption{Method Performance in Fake News Detection.}
\label{tab:general_performance}
\begin{adjustbox}{max width=\columnwidth}
\begin{tabular}{lcccc}
\toprule[1pt]
\multirow{2}{*}{\textbf{}} & \multicolumn{2}{c}{\textbf{ReCOVery}} & \multicolumn{2}{c}{\textbf{MM-COVID}} \\ \cline{2-5}
 & \textbf{Macro F1} & \textbf{Micro F1} & \textbf{Macro F1} & \textbf{Micro F1} \\ \midrule[0.5pt]
\textbf{LIWC} & .500 $\pm$.034 & .612 $\pm$.032 & .704 $\pm$.033 & .707 $\pm$.032 \\
\textbf{RST} & .514 $\pm$.035 & .626 $\pm$.030 & .529 $\pm$.031 & .608 $\pm$.028 \\
\textbf{T2VW} & .619 $\pm$.035 & .709 $\pm$.032 & .781 $\pm$.034 & .782 $\pm$.034 \\
\textbf{T2VC} & .622 $\pm$.033 & .712 $\pm$.027 & .769 $\pm$.031 & .771 $\pm$.031 \\
\textbf{NewsTag\textbackslash{}W} & .673 $\pm$.036 & .772 $\pm$.027 & .835 $\pm$.019 & .836 $\pm$.018 \\
\textbf{NewsTag\textbackslash{}I} & .699 $\pm$.024 & .772 $\pm$.026 & .827 $\pm$.032 & .828 $\pm$.032 \\
\textbf{NewsTag} & \textbf{.727 $\pm$.022} & \textbf{.806 $\pm$.016} & \textbf{.854 $\pm$.027} & \textbf{.855 $\pm$.027} \\
\bottomrule[1pt]
\end{tabular}
\end{adjustbox}
\end{table}

\begin{figure}[t]
    \centering
    \begin{minipage}{\columnwidth}
    \subfigure[$k_1=1$]{
    \begin{minipage}{0.48\columnwidth}
    \includegraphics[width=\columnwidth]{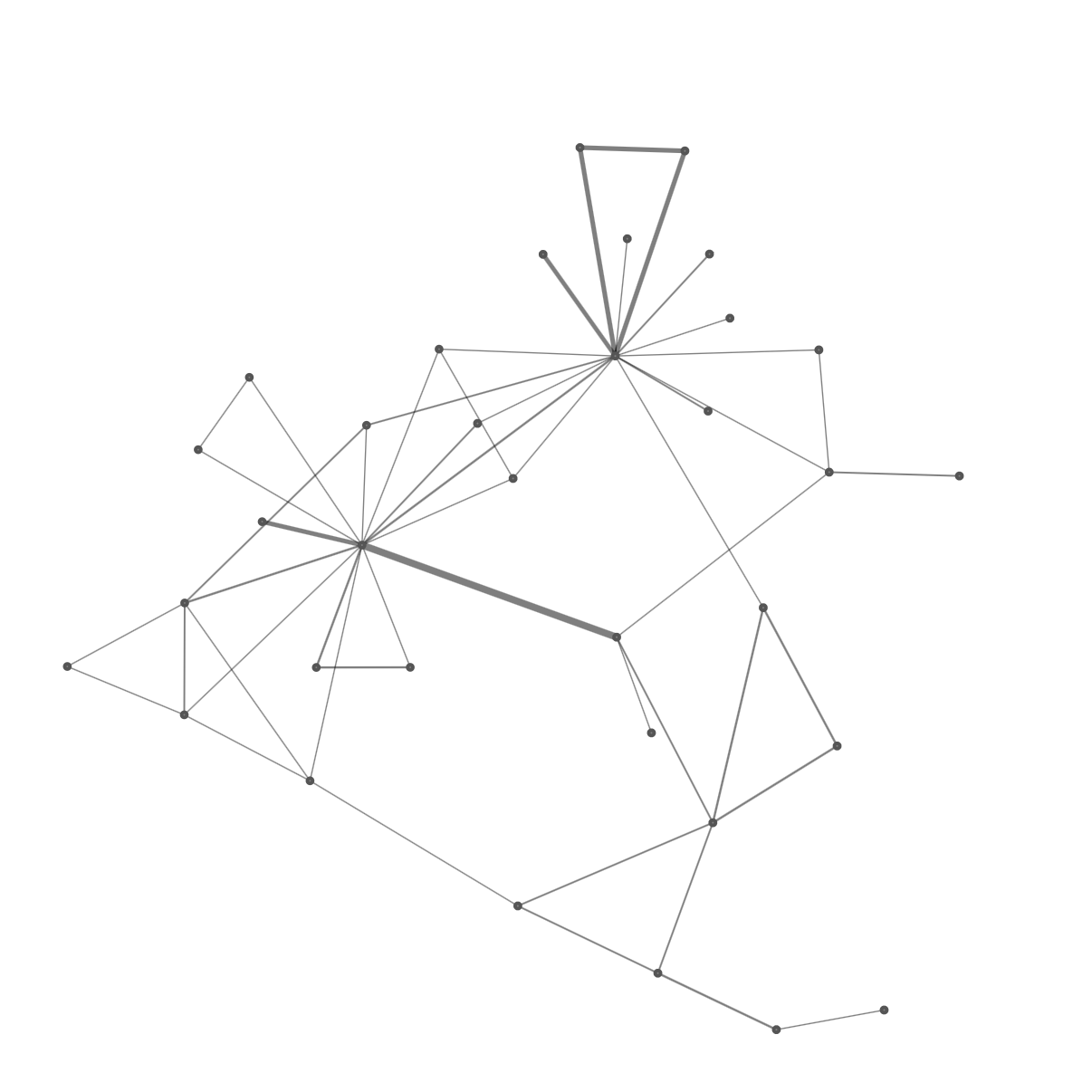}
    \end{minipage}}
    \subfigure[$k_1=2$]{
    \begin{minipage}{0.48\columnwidth}
    \includegraphics[width=\columnwidth]{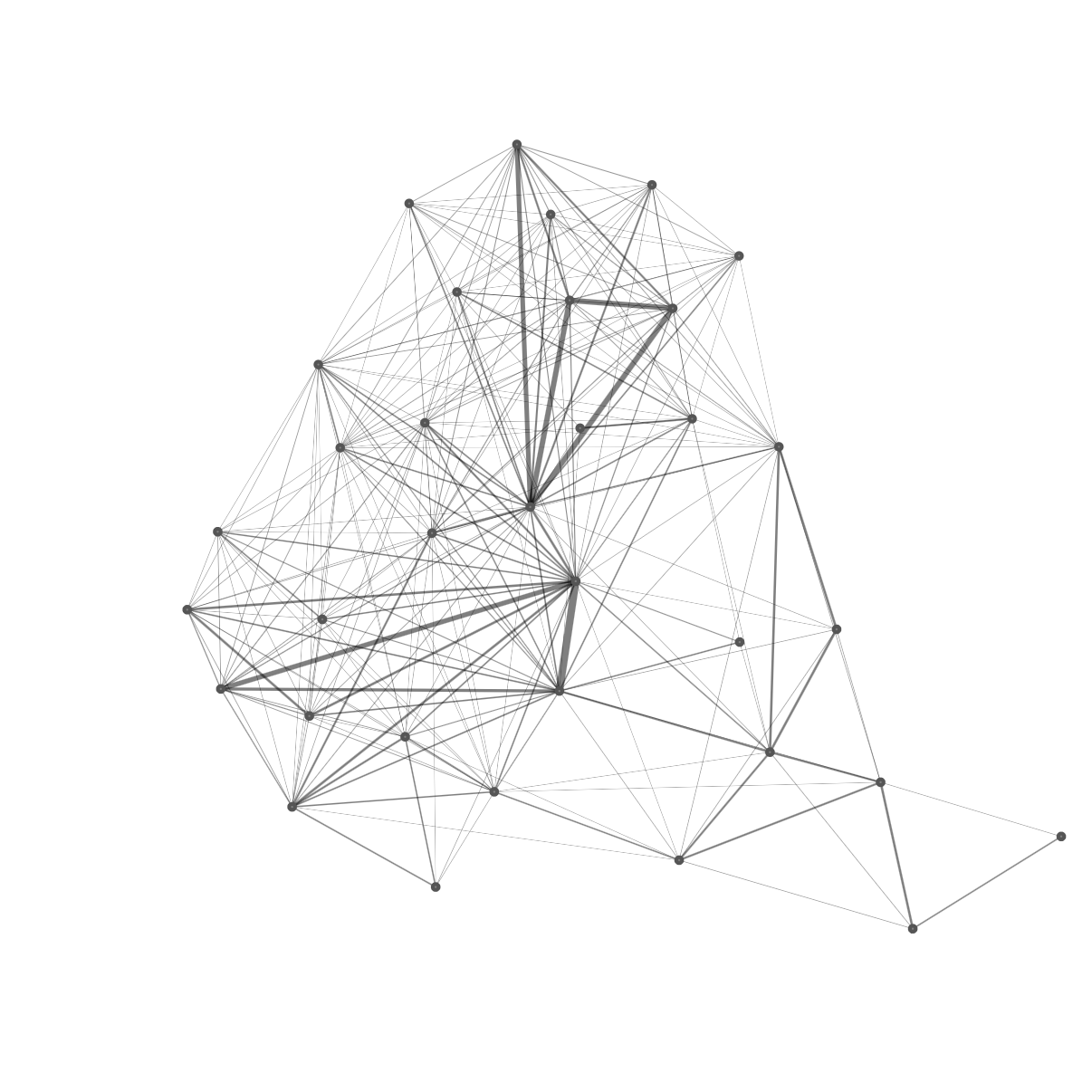}
    \end{minipage}}
    \end{minipage}
    \caption{An Illustrated Graph. (a) only contains direct relation among nodes and (b) includes direct and indirect relationships among the same set of nodes.}
    \label{fig:evolutional_hashtag_graph}
\end{figure}

\begin{figure*}[t]
    \centering
    \includegraphics[width=0.96\columnwidth]{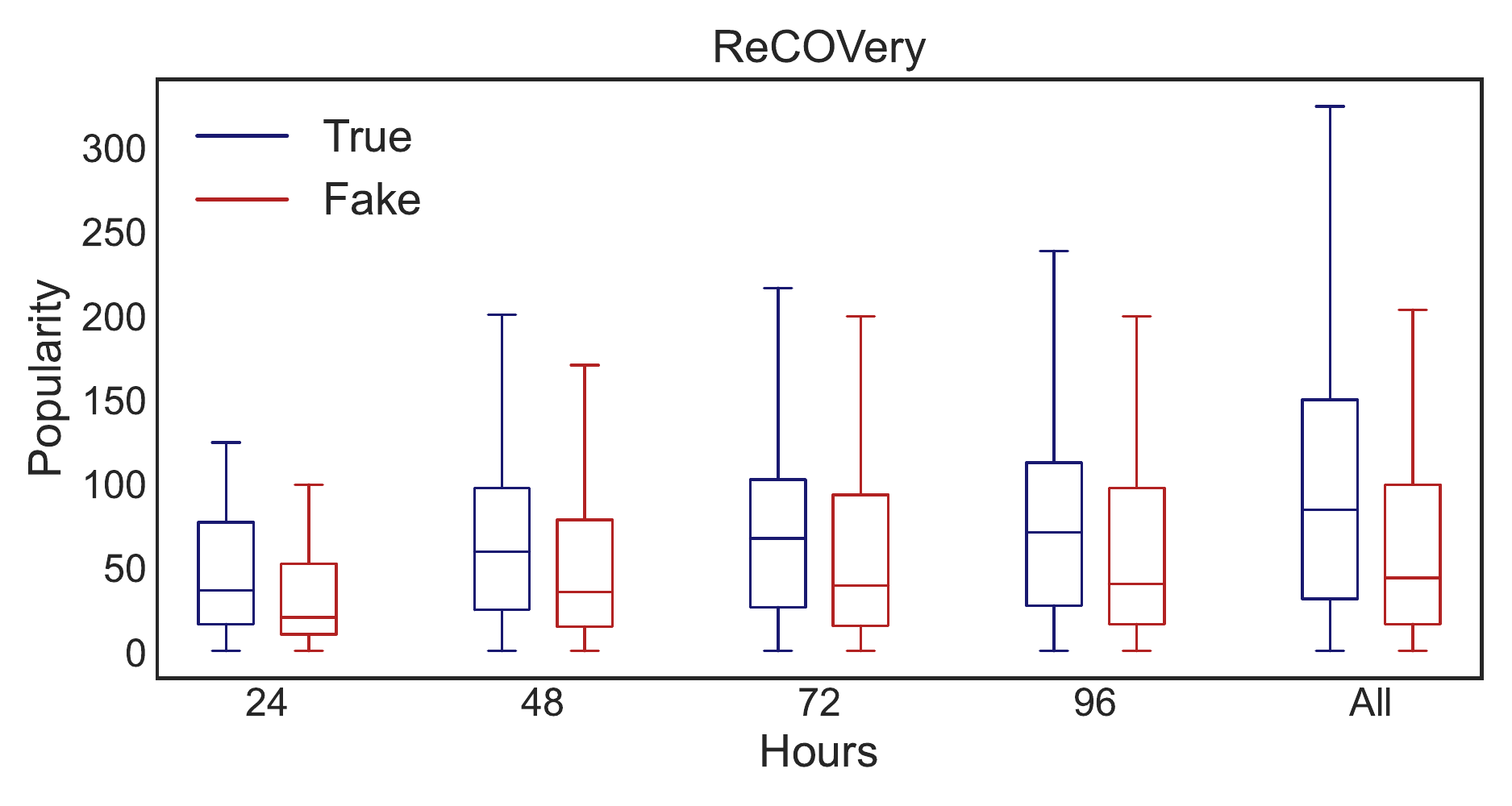}\qquad
    \includegraphics[width=0.96\columnwidth]{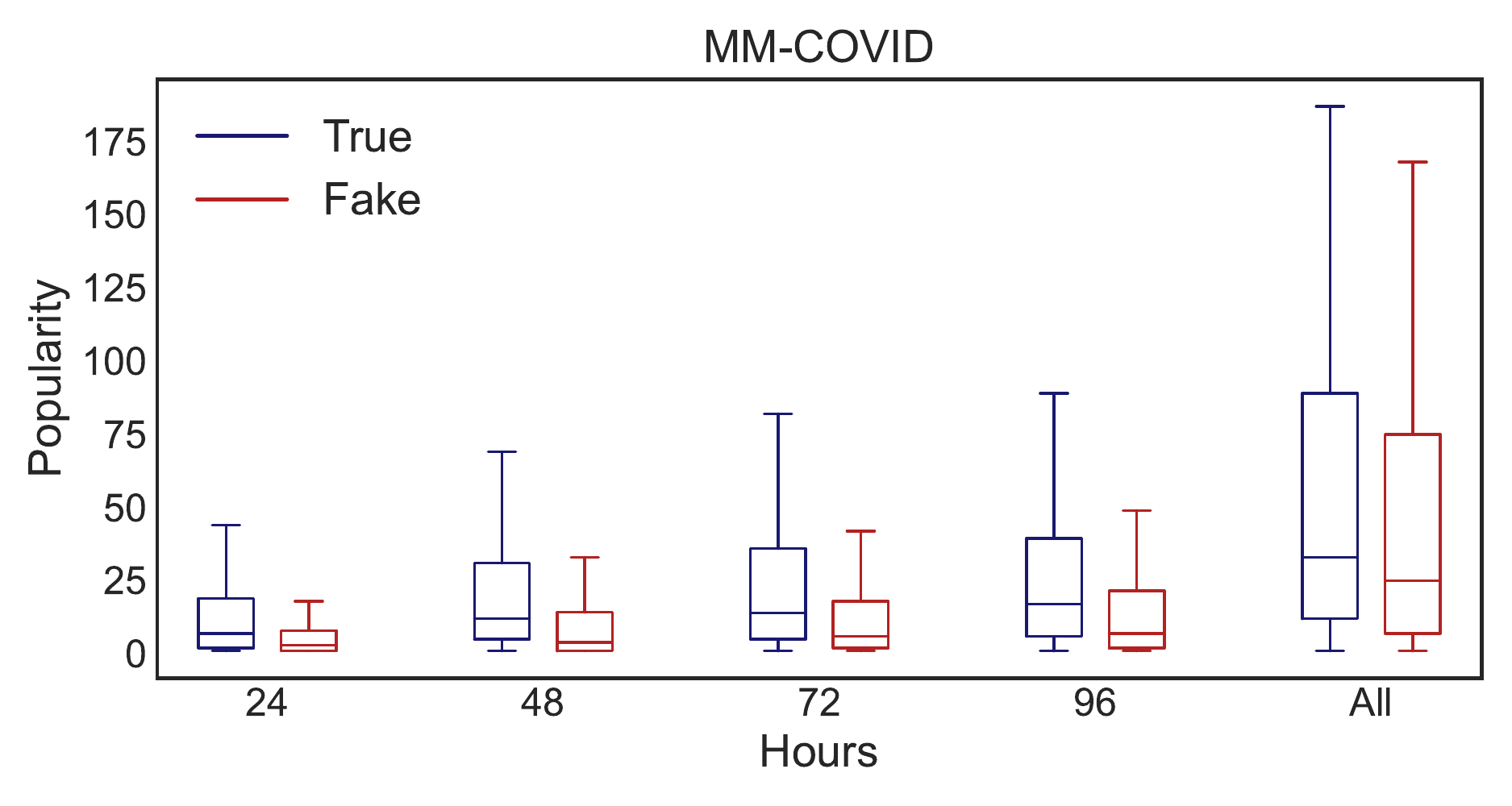}
    \caption{Differential Propagation of Fake and True News. We record the overall popularity of each news article or statement within fixed hours after it was published. Results indicate that the popularity of true news is higher than that of fake news.}
    \label{fig:news_patterns}
\end{figure*}

\begin{figure}
    \subfigure[ReCOVery]{
        \centering
        \includegraphics[width=0.5\columnwidth]{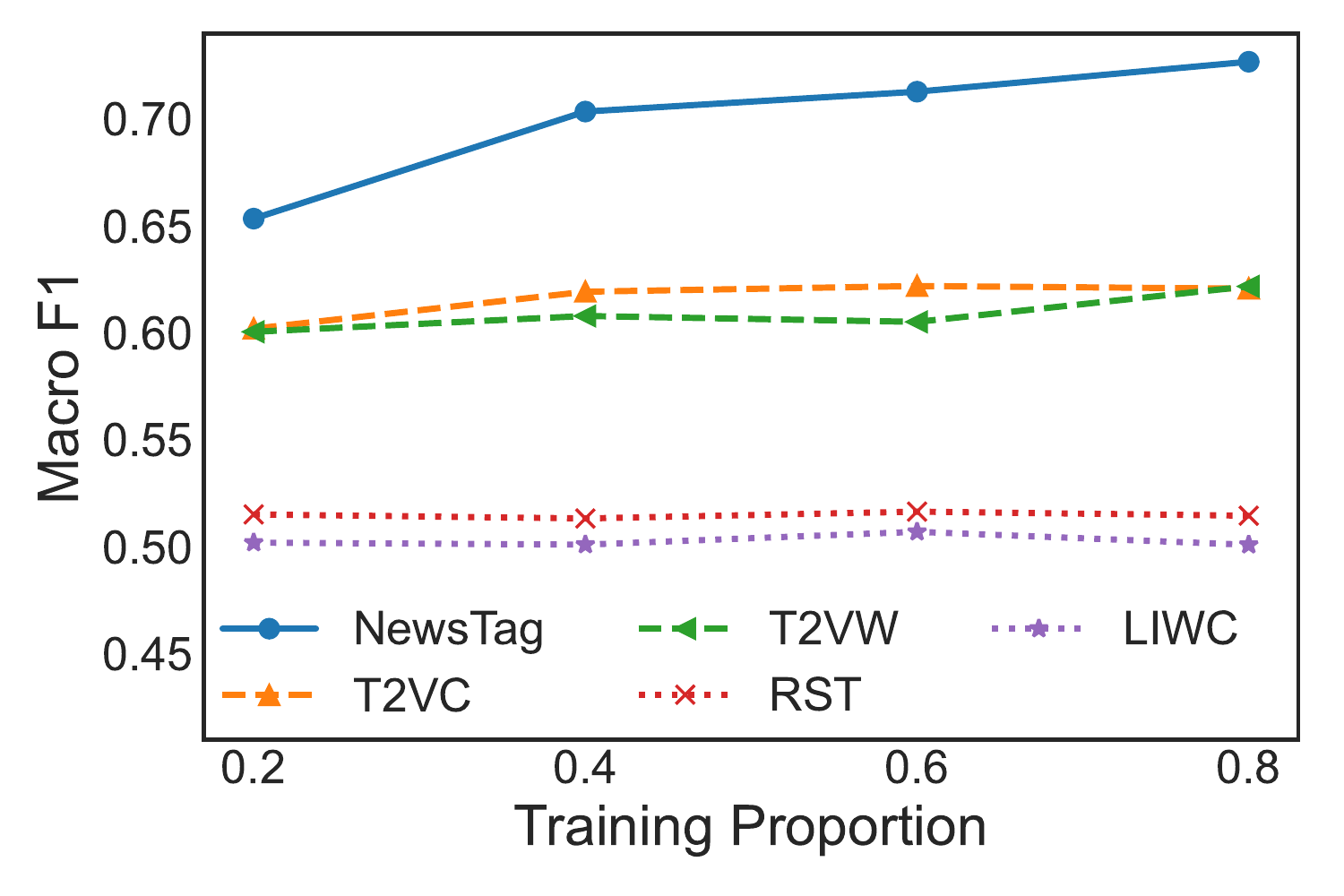}
        \includegraphics[width=0.5\columnwidth]{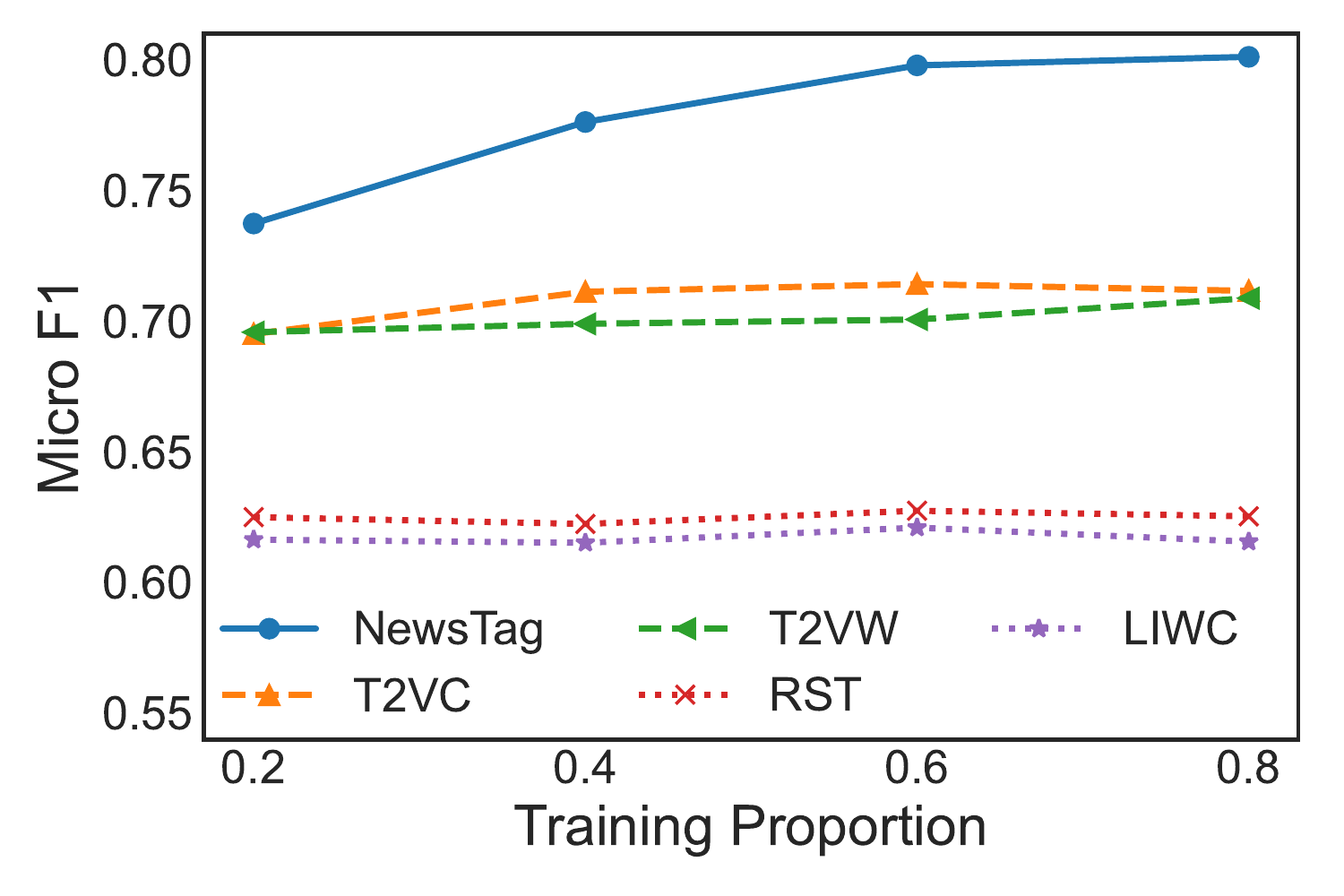}
    }
    \subfigure[MM-COVID]{
        \centering
        \includegraphics[width=0.5\columnwidth]{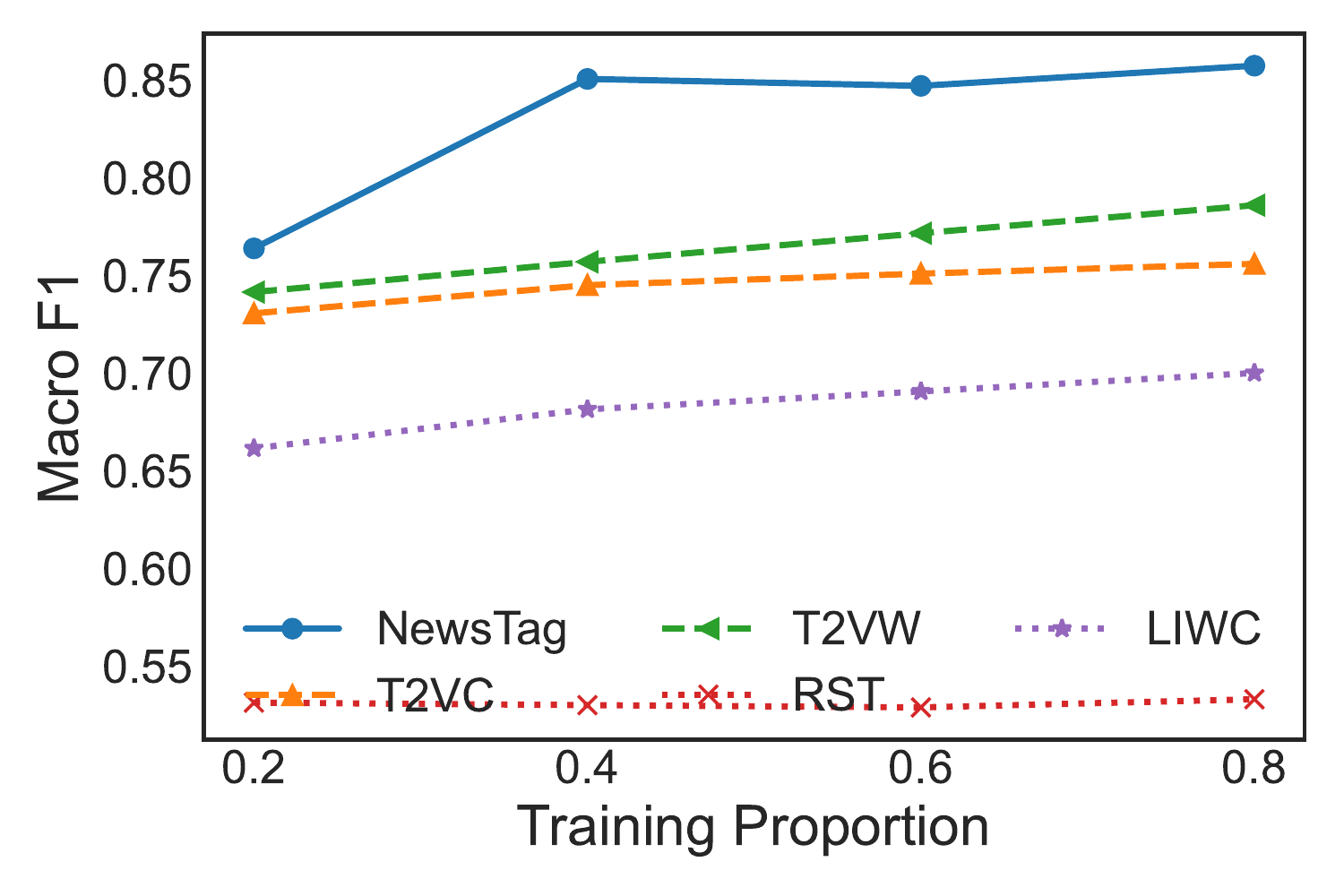}
        \includegraphics[width=0.5\columnwidth]{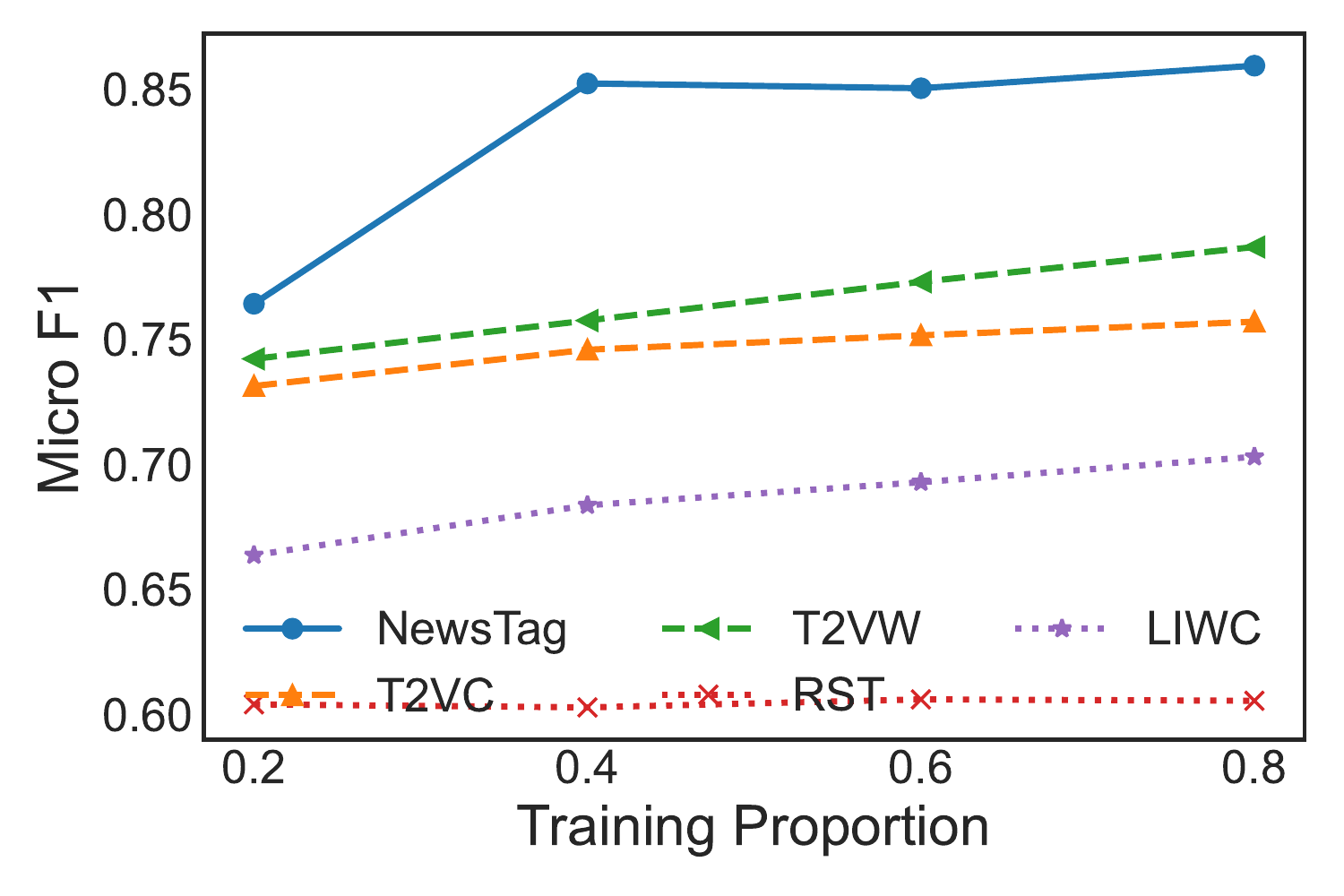}
    }
    \caption{Impact of the Volume of Training Data on Method Performance. We changed the proportion of data used for training from 20\% of the overall data to 80\%. Results indicate that the proposed method outperforms baselines.}
    \label{fig:volume}
\end{figure}

\begin{figure}
    \subfigure[ReCOVery]{
        \centering
        \includegraphics[width=0.48\columnwidth]{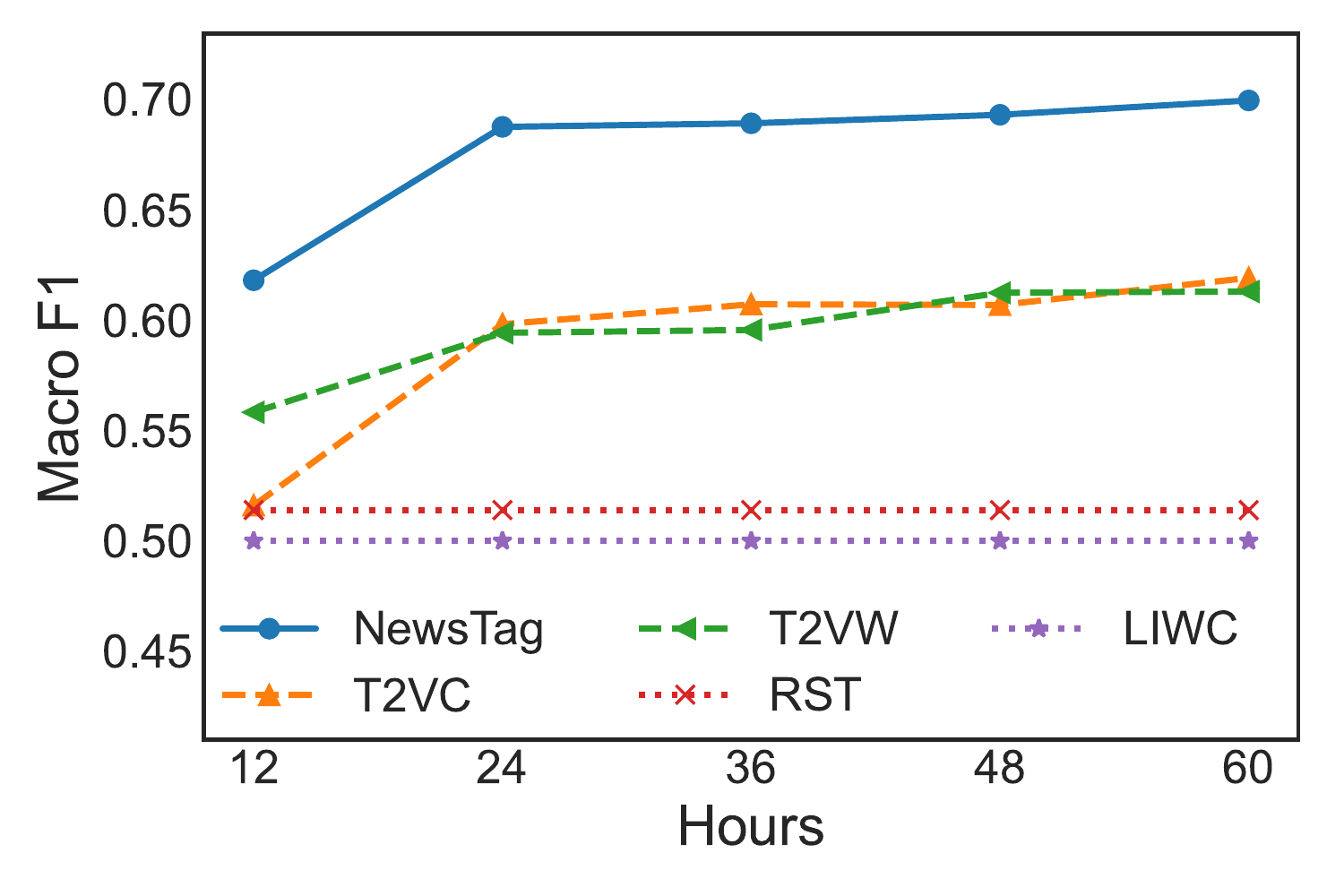}
        \includegraphics[width=0.48\columnwidth]{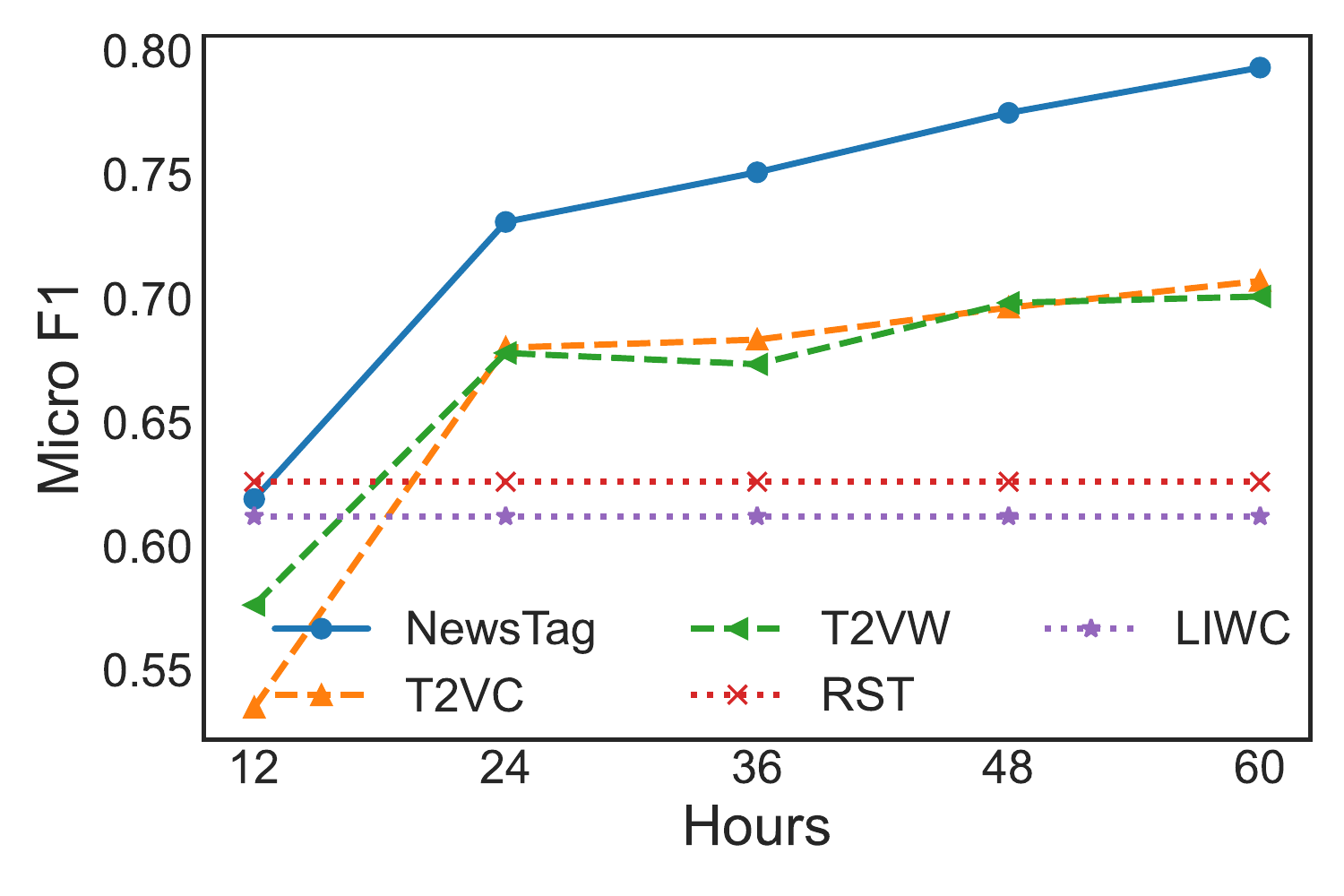}
    }
    \subfigure[MM-COVID]{
        \centering
        \includegraphics[width=0.48\columnwidth]{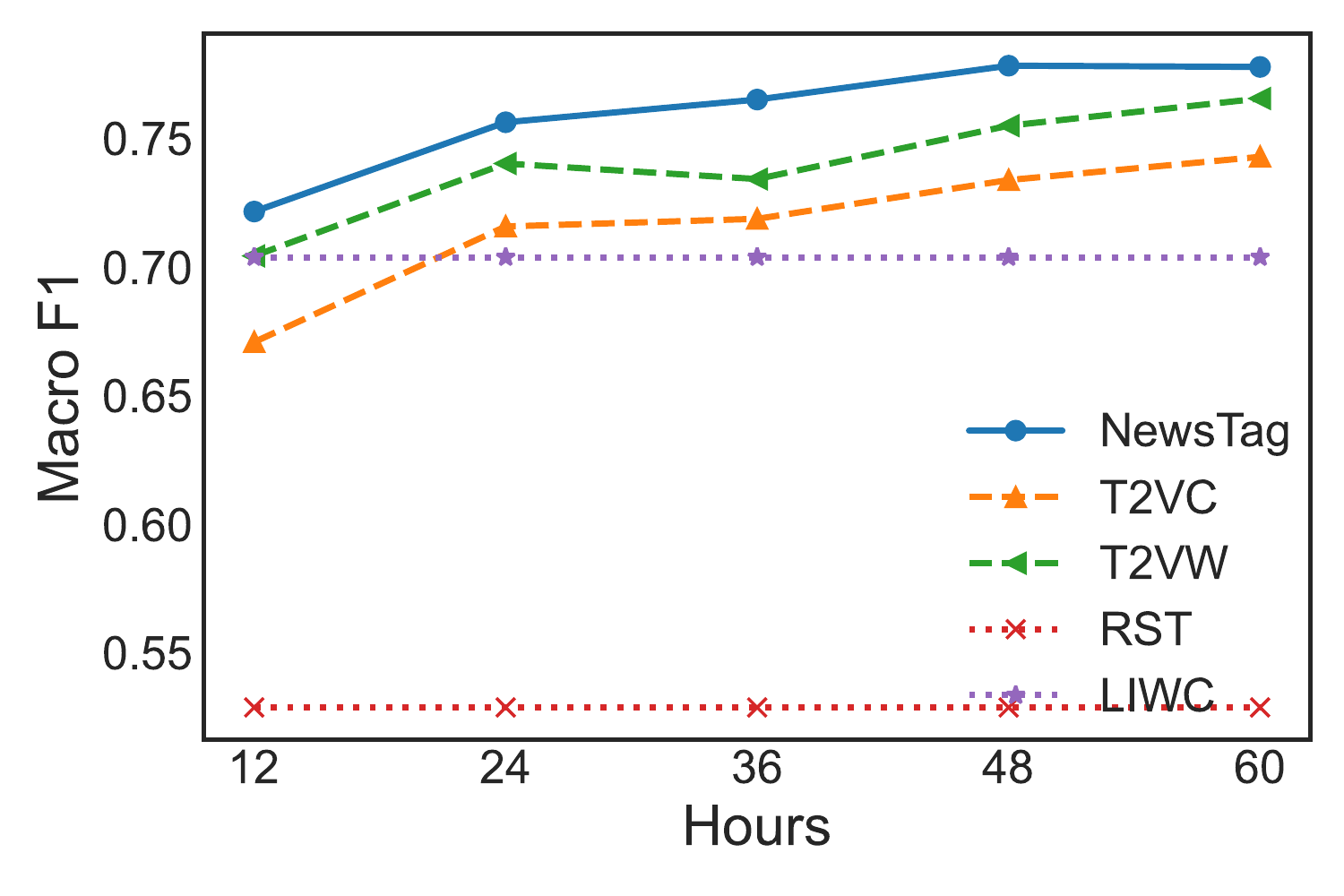}
        \includegraphics[width=0.48\columnwidth]{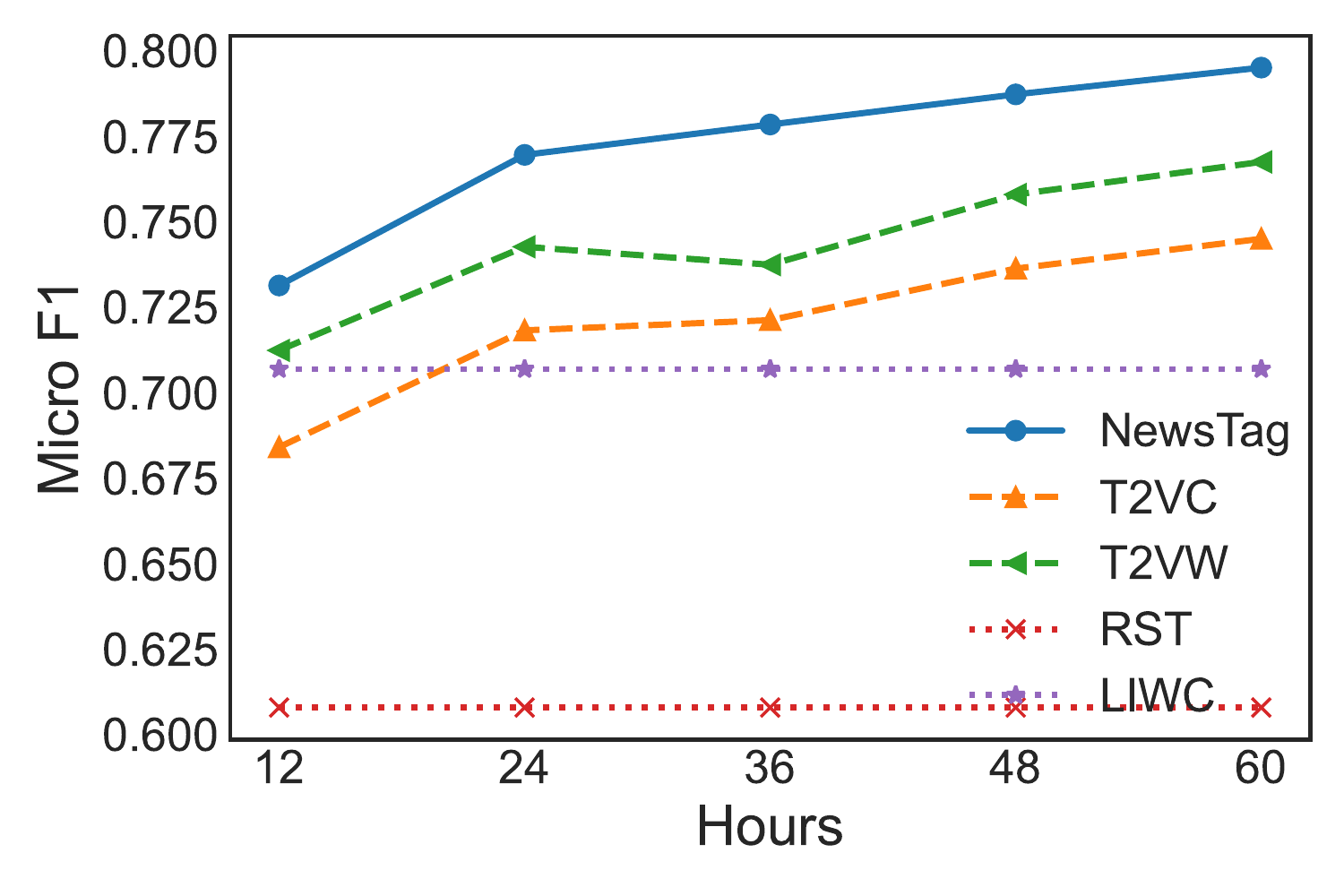}
    }
    \caption{Impact of the Limited Time of Detection on Method Performance. We vary the detecting time of methods from 12 hours after each news article or statement is published online to 60 hours. Results indicate that the proposed method can outperform baselines.}
    \label{fig:hour}
\end{figure}

\subsubsection{Fake News Early Detection}
\label{subsubsec:early_detection}

Identifying fake news at an early stage of spreading is essential as individuals tend to trust more those fake news articles or statements that have been widely spread, i.e., \textit{validity effect}~\cite{hogg2020social}. Therefore, we evaluate the method of predicting fake news by defining the ``early stage'' from two perspectives: 
\begin{itemize}
    \item The stage at which the volume of data available for learning is limited; and 
    \item The limited time a news article or statement can diffuse on social media after being published online.
\end{itemize} 

\paragraph{I. Impact of the volume of training data}
To evaluate the impact of the volume of training data on the proposed method, we changed the proportion of data used for training from 20\% of the overall data to 80\%. Figure~\ref{fig:volume} presents the corresponding method performance. Results indicate that when increasing the training proportion from 20\% to 80\%, the proposed method can achieve a macro F1 score ranging from $\sim$0.66 to $\sim$0.73 and a micro F1 score ranging from $\sim$0.74 to $\sim$0.80 based on the ReCOVery dataset; for MM-COVID data, both F1 scores range from $\sim$0.77 to $\sim$0.86. The method consistently outperforms the baseline T2VC and T2VW and significantly outperforms the baseline LIWC and RST on both datasets. Specifically, trained with 20\% ReCOVery data, the proposed method increases both F1 scores by $\sim$5\% compared to T2VC and T2VW and by more than 10\% compared to LIWC and RST. For MM-COVID data, such improvement is $\sim$3\% for T2VC and T2VW, $\sim$10\% for LIWC, and more than 15\% for RST. 

\begin{figure*}[t]
    \subfigure[ReCOVery]{
        \centering
        \includegraphics[width=\columnwidth]{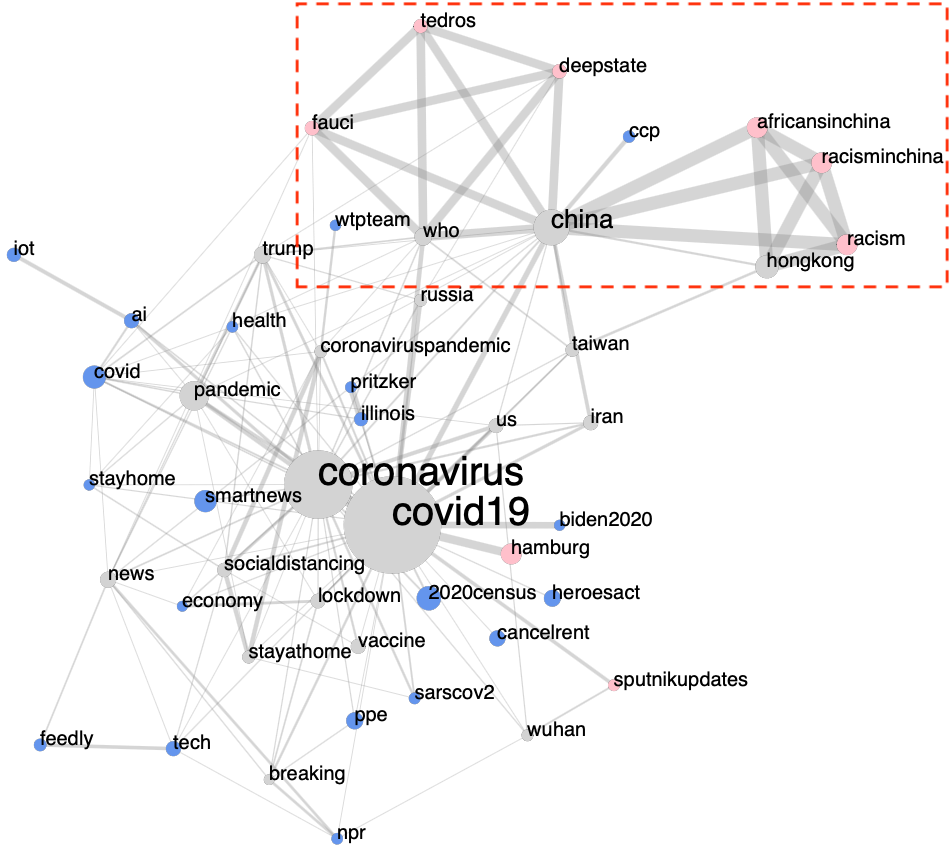}
    }
    \subfigure[MM-COVID]{
        \centering
        \includegraphics[width=\columnwidth]{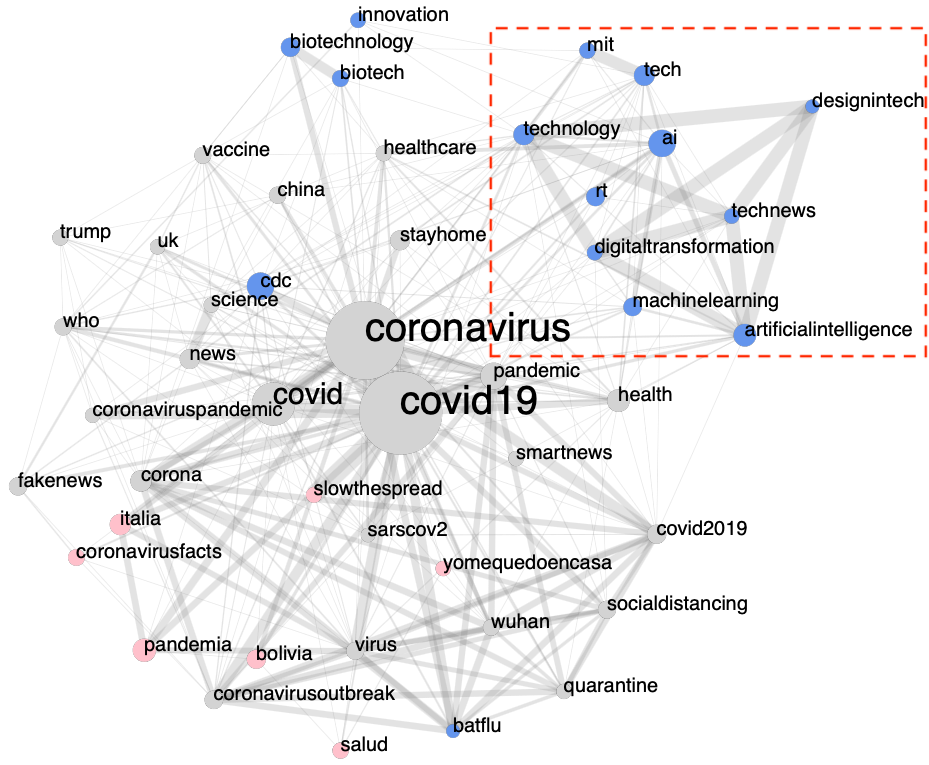}
    }
    \caption{Hashtag Graph. 
    Blue indicates that $\mathbf{c_*}_{k} \geq 0.9$, red indicates $\mathbf{c_*}_{k} \leq -0.9$, and otherwise gray. }
    \label{fig:hashtag_graph}
\end{figure*}

\paragraph{II. Impact of the limited time of detection}

We vary the detecting time of methods from 12 hours after each news article or statement is published online to using all propagation information. Figure~\ref{fig:hour} presents the corresponding method performance. Results indicate that when delaying the detecting time from 12 hours to 60 hours, the proposed method can achieve a macro F1 score ranging from $\sim$0.62 to $\sim$0.70 and a micro F1 score ranging from $\sim$0.62 to $\sim$0.79 based on the ReCOVery dataset; for MM-COVID data, both F1 scores range from $\sim$0.73 to $\sim$0.79. The method consistently outperforms propagation-based baselines (T2VC and T2VW) and outperforms content-based baselines (LIWC and RST) most of the time. 
Specifically, with 12 hours delay in the time of detection, the proposed method increases both F1 scores by $\sim$5\% compared to T2VW and $\sim$10\% compared to T2VC using ReCOVery data; compared to LIWC and RST, it enhances macro F1 score by $\sim$5\% and is comparable in micro F1 score.
For MM-COVID data, the improvement in both F1 scores is $\sim$2\% for T2VC and LIWC, $\sim$5\% for T2VW, and more than 10\% for RST. 

\subsubsection{Case Study}
\label{subsubsec:case_study}

We visualize a part, for better visualization, of the proposed hashtag graph for both datasets in Figure \ref{fig:hashtag_graph}. For each dataset, we color nodes (hashtags) based on their ``credibility'' obtained using all the data (i.e., $\mathbf{c_*}$). Note that the difference between $\mathbf{c_*}$ and $\mathbf{c_0}$ is that the former is computed by all the data while the latter is merely based on the training data. Below are two main observations that we have.

First, we observe a typical subgraph from both hashtag graphs (see the marked upper-right area of Figure~\ref{fig:hashtag_graph}). These hashtags within the subgraph are connected with the most significant weights, i.e., they are closely related. Meanwhile, these hashtags all have a relatively consistent ``credibility'' score. For MM-COVID data, most of them share a remarkably high credibility score ($\mathbf{c_*}_{k} \geq 0.9$), i.e., they are primarily involved in the spread of true news. For ReCOVery data, they all have an extremely low credibility score ($\mathbf{c_*}_{k} \leq -0.9$), i.e., they are mostly involved in the spread of fake news. This discovery validates our assumption when defining the cost function that hashtags closely related share a similar ``credibility'' score.

Furthermore, we observe that hashtags related to COVID-19 conspiracy theories often exhibit actual ($\mathbf{c_*}$) and estimated ($\mathbf{\hat{c}}$) low ``credibility'' scores (see Table \ref{tab:conspiracy_theory}), i.e., they are mainly involved in the spread of fake news. This observation further validates the approach taken by our method. 

\section{Related Work}
\label{sec:review}

Current methods of evaluating news credibility can be generally grouped into two categories: (I) content-based and (II) propagation-based methods.

\paragraph{I. Content-based methods}
Within a traditional machine learning framework, content-based methods manually extract a group of features from the news content to predict news credibility. Such features can be linguistic (e.g., personal pronouns), psychological (e.g., emotions), and statistical features (e.g., the number of words) \cite{perez2018automatic,potthast2018stylometric,przybyla2020capturing}. They can also be discourse-level features, such as rhetorical relations among content sentences or phrases~\cite{karimi2019learning,rubin2015towards}. With the rapid development of neural networks, increasing deep models have emerged to learn the representation of news content and performed well in predicting fake news~\cite{zhang2020fakedetector,rashkin2017truth,bal2020analysing,wang2018eann}. However, these methods often demand massive training data, while annotating news credibility is time-consuming and requires domain knowledge. 
Content-based methods have the natural advantage of detecting fake news before the news article or statement has been propagated on social networks. However, such methods can be vulnerable to malicious news writers' manipulation of their writing style, which motivates the development of propagation-based methods.

\begin{table}[t]
    \centering
    \captionof{table}{Hashtags Related to Conspiracy Theories.}
    \label{tab:conspiracy_theory}
    \begin{tabular}{clrr}
    \toprule[1pt]
    \textbf{Topic} & \textbf{Hashtag} & $\mathbf{c_*}$ & $\mathbf{\hat{c}}$\protect\footnotemark[2] \\ \midrule[0.5pt]
    {Vaccine} & \#saynotovaccine & -1.0 & -1.0 \\
     & \#nomorevaccines & -1.0 & -1.0 \\
     & \#novaccineforme & -1.0 & -0.6 \\
    {5G} & \#5gdeath & -1.0 & -0.6 \\ 
    {Deep State} & \#deepstate & -0.9 & -0.7 \\ 
    {Bioweapon} & \#bioweapon & -1.0 & -1.0 \\ 
    {Plandemic} & \#plandemic & -1.0 & -1.0 \\ 
    \bottomrule[1pt]
    \end{tabular}
\end{table}
\footnotetext[2]{Note that the entry value of $\mathbf{\hat{c}}$ is shrank by $\mu$. $\mathbf{\hat{c}}\in [-1,1]^q$ if $\mu=0$ and $\mathbf{\hat{c}}=\mathbf{0}$ if $\mu=1$. We use the best $\mu$ (0.4), and map $\mathbf{\hat{c}}$ to $[-1,1]^q$ for a better comparison with $\mathbf{c_*}$ that is $\in [-1,1]^q$ as well.}

\paragraph{II. Propagation-based methods}
Propagation-based methods predict news credibility by mining social context information. Existing approaches have developed social connections~\cite{wu2018tracing}, user profiles~\cite{cheng2021causal,shu2019role}, and temporal patterns in news spreading~\cite{vosoughi2018spread,sharma2021network,ma2018rumor,budak2019happened}. They have also exploited the available content, including user emotions and stances within posted tweets, retweets, and replies~\cite{castillo2011information,jin2016news,lukasik2019gaussian}. 
These methods have contributed significantly to news credibility evaluation, a few of which have also used hashtag information, e.g., how many hashtags in user tweets~\cite{castillo2011information}. Nevertheless, hashtags, especially their \textit{relationships}, have rarely been investigated to address the problem, whose value has been specified in Section \ref{sec:intro}.

\section{Conclusion}
\label{sec:conclusion}

We propose a language-independent semi-supervised approach for fake news detection, comprehensively exploring the direct and indirect relationships among hashtags. The method of computing indirect relations among nodes can be applied to most homogeneous graphs. Experiments are conducted on two COVID-19 news datasets, ReCOVery and MM-COVID. Results validate the value of hashtags and the proposed method's effectiveness, especially in the early detection of fake news. In our future studies, we will work on capturing the indirect relationship among \textit{heterogeneous} nodes. It allows us to extend our homogeneous graph to be heterogeneous by further integrating more information (e.g., news content) to enhance the method performance of news credibility evaluation.

\section*{Ethical Considerations}

In this work, we propose a new method to identify fake news. This work aims to contribute to society by identifying fake news early in its dissemination on social media. Such early identification can mitigate the negative impacts caused by fake news on, e.g., democracies, economies, and public health specified in the introduction. The proposed method investigates relationships among news articles (or statements) and hashtags of posts involved in their dissemination on social media, which does not require and use the data revealing user privacy. We verify the proposed method in two public datasets. Note that two datasets only release the IDs of the collected Twitter data for non-commercial research use to comply with Twitter's Terms of Service. We will release our code and data publicly once the paper is accepted by adhering to this rule.  

% \balance

%%
%% The acknowledgments section is defined using the "acks" environment
%% (and NOT an unnumbered section). This ensures the proper
%% identification of the section in the article metadata, and the
%% consistent spelling of the heading.
% \begin{acks}

% \end{acks}

%%
%% The next two lines define the bibliography style to be used, and
%% the bibliography file.
\bibliographystyle{ACM-Reference-Format}
\bibliography{references}

%%
%% If your work has an appendix, this is the place to put it.
% \appendix

\end{document}